\documentclass[10pt,twocolumn,twoside]{IEEEtran}
\IEEEoverridecommandlockouts
\usepackage{amsmath,amssymb,graphics,epsfig,subfigure,adjustbox,bbm,color,xcolor}
\usepackage{sgamevar,mathtools}
\usepackage{float}
\let\labelindent\relax
\usepackage{enumitem}
\usepackage{comment}
\usepackage{algpseudocode,algorithm,verbatim}

\newcommand{\mc}[1]{\mathcal{#1}}
\newcommand{\te}[1]{\mathrm{#1}}
\newcommand{\field}[1]{{\mathbb{#1}}}
\newcommand{\be}{\begin{equation}}
\newcommand{\ee}{\end{equation}}
\newcommand{\bea}{\begin{eqnarray}}
\newcommand{\eea}{\end{eqnarray}}
\newcommand{\ba}{\begin{array}}
\newcommand{\ea}{\end{array}}
\newcommand{\beas}{\begin{eqnarray*}}
\newcommand{\eeas}{\end{eqnarray*}}
\newcommand{\leftm}{\left[\begin{array}}
\newcommand{\rightm}{\end{array}\right]}

\newcommand{\bs}{\setminus}
\newtheorem{thm}{\bf Theorem}[section]

\newtheorem{prop}{\bf Proposition}

\newtheorem{definition}{\it Definition}[section]

\DeclareMathOperator*{\argmax}{\arg\!\max}

\title{Maximizing User Engagement in Social Networks: A Game-Theoretic Approach to Network Participation and Resource Sharing}
\author{Ahmed Luqman and  Hassan Jaleel
\thanks{A. Luqman and H. Jaleel are with the Intelligent Machines \& Sociotechnical Systems (iMaSS) Lab, Department of Electrical Engineering, Syed Babar Ali School of Science \& Engineering at LUMS, Lahore, Pakistan. Emails: \texttt{hassan.jaleel@lums.edu.pk, 24100041@lums.edu.pk}}}
\begin{document}
\maketitle
\begin{abstract}
We propose a game-theoretic framework to model and optimize user engagement in cooperative activities over social networks. While traditional diffusion models suggest that individuals are only influenced by their neighbors, empirical evidence shows that diffusion alone does not fully explain network evolution, and non-diffusion factors play a significant role in network growth. We model network participation and resource-sharing as strategic games involving boundedly rational players to address this gap between the analytical models and empirical evidence. Specifically, we employ Log-Linear Learning (LLL), a version of noisy best response, to capture players' decision-making strategies. By incorporating stochastic decision models like LLL, our framework integrates both diffusion and non-diffusion dynamics into network evolution dynamics. Through equilibrium analysis and simulations, we demonstrate that our model aligns with theoretical predictions from existing analytical frameworks and empirical observations across various initial network configurations. Our second contribution is a novel method for selecting anchor nodes to enhance user participation. This approach allows system designers to identify anchor nodes and compute their incentives in real time under a more realistic information requirement constraints as compared to the existing approaches. The proposed approach adapts to changing network conditions by reallocating resources from less impactful to more influential nodes.  Furthermore, the method is resilient to anchor node failures, ensuring sustained and continuous network participation.

\end{abstract}

\section{Introduction}
%

%
Researchers have long been interested in understanding the formation and evolution of communities within social networks (\cite{bala2000noncooperative}, \cite{granovetter1978threshold},  \cite{goldenberg2001using}, \cite{rand2013human}, \cite{simpson2015beyond}, \cite{skyrms2009dynamic},  and \cite{leskovec2008microscopic}). In particular, what factors drive individuals to participate in any network activity (\cite{granovetter1978threshold}, \cite{goldenberg2001talk}, and \cite{burke2009feed})? What are the structural properties of the underlying social networks that influence individual decisions at the microscopic level \cite{leskovec2008microscopic} and the resulting network properties at the macroscopic level \cite{leskovec2005graphs}? How many individuals will eventually participate in a network activity given an underlying social network? While such problems have been a focus of research for a long time, the emergence of online social networks has introduced new dimensions to this research with the availability of large data sets of network formation (see \cite{borgatti2018analyzing}, \cite{jackson2008social}, \cite{lattanzi2009affiliation}, and \cite{backstrom2006group} for online social networks). 

We present a framework for modeling and analyzing the evolution of groups and communities in social networks by formulating interactions among neighboring users as a non-cooperative game. We explore two motivating scenarios. The first is a \emph{network participation problem}, where individuals in a social network must decide whether to engage in network activity or drop out. Participation provides certain benefits but also incurs costs, which could be financial (e.g., subscription fees) or non-financial, such as time commitment. The second scenario is a \emph{network resource-sharing problem}, where each user needs access to a specific set of resources but only possesses a smaller subset. However, the remaining resources may be available through their neighbors in the social network. As a result, individuals must decide whether to share their resources with neighbors in exchange for access to the resources their neighbors possess.

Modeling and analyzing the evolution of user engagement in network activities across social networks has been an active area of research (see, for instance, \cite{granovetter1978threshold}, \cite{goldenberg2001talk}, \cite{wu2013arrival}, \cite{backstrom2006group},  \cite{burke2009feed}, \cite{bhawalkar2015preventing},  \cite{malliaros2013stay}, \cite{garcia2013social}, and \cite{kempe2005influential}). A widely adopted model for user engagement in these frameworks is the \emph{threshold model} from \cite{granovetter1978threshold}, which posits that an individual's decision to participate in or withdraw from a network activity is influenced by the number of their friends participating in the activity. Specifically, individuals join an activity if the number of their participating friends exceeds a certain threshold. A common assumption in this model is to set a uniform threshold value $k$ for all users, leading to the prediction that the $k$-core of the social network will represent the subgraph of participating players.

The threshold model predicts that user participation diffuses through the network, starting from an initial group of participants. The same is true for the \emph{independent cascade model} in \cite{goldenberg2001using} and \cite{goldenberg2001talk}, which is based on the theory of interacting particle systems \cite{liggett1985interacting}. In this probabilistic model, active users try to convince their neighbors to join. An issue with diffusion-based models is that multiple empirical studies on online social networks have revealed that the evolution of user participation also involves non-diffusion components \cite{backstrom2006group, kairam2012life, malliaros2013stay, garcia2013social}. Non-diffusion growth occurs when isolated users without prior friends join randomly across the network. As highlighted in \cite{kairam2012life}, non-diffusion growth is crucial for understanding the steady state of the engagement network since pure diffusion often limits network evolution. Thus, a limitation of the threshold  

Another issue with diffusion-based models, such as the threshold model, is a limitation in their analysis approach. The steady-state analysis of these models often begins with the assumption that all users are initially engaged \cite{bhawalkar2015preventing}.  Then, players with fewer than $k$ neighbors withdraw from the network, triggering a cascade of withdrawals that eventually stabilizes at the $k$-core. However, assuming that all players are initially engaged in a network activity is both restrictive and unrealistic, raising concerns about the $k$-core’s validity as a steady-state predictor for active participants. A similar approach was presented in \cite{abbas2019graph} for the resource-sharing problem, where users with a subset of $s$ resources would opt out if they did not have access to all $r$ resources through their neighbors. This disengagement would trigger a chain reaction of withdrawals that would ultimately converge to the $(r, s)$-core.


To address these limitations, {we present a modeling framework that integrates both diffusion and non-diffusion dynamics to capture the evolution of user engagement in a network activity. Additionally, we introduce a novel approach for the steady-state analysis of the engagement network that is independent of the initial conditions}. Our framework employs a strategic game formulation for network participation and resource-sharing problems. We consider players with bounded rationality and propose Log-Linear Learning (LLL) for modeling players' decision strategies. Log-linear learning (LLL), as presented in \cite{Blume93}, \cite{Marden12}, and \cite{alos10}, is a variant of noisy best response dynamics where players predominantly exploit the best response but occasionally explore other actions. These exploration and exploitation behaviors lead to diffusion and non-diffusion network growth, respectively, which justifies our choice of LLL for modeling user decision strategy. We also present a stochastic stability analysis for the network participation game using the Radius-CoRadius approach from \cite{alos10}.


Our second contribution is a novel framework for maximizing user participation in network activities. We introduce a system designer as a Principal-Agent (PA) who can incentivize individuals to join the network. While the concept of incentivizing influential users, known as anchor nodes, has been extensively studied (e.g., \cite{bhawalkar2015preventing}, \cite{kempe2005influential}, \cite{abbas2019graph}, \cite{medya2019k}, \cite{zhang2017finding}, \cite{linghu2020global}), existing approaches typically rely on greedy heuristics to identify a set of anchor nodes before the network activity starts. Once identified, these players permanently serve as anchor nodes and receive incentives. However, this approach assumes that the Principal-Agent (PA) has complete knowledge of the underlying social network structure, which can be a restrictive assumption. Additionally, the method lacks resilience to anchor node failures, as even a small number of failures would require the PA to recompute the entire anchor set. To address these limitations, our framework continuously reassesses the network configuration at each decision interval for selecting anchor nodes in real-time. The proposed method only requires information about players who have already committed to network participation, which is a more realistic assumption for the PA. Furthermore, by dynamically selecting anchor nodes based on their current potential to contribute, our approach enhances resilience against anchor node failures.



We presented a preliminary version of these ideas in \cite{luqman2024dynamic}, where we introduced the strategic game formulation for network participation and resource-sharing games. In this paper, we have significantly extended the analysis by incorporating stochastic stability analysis using the Radius-CoRadius approach. Additionally, the problem setup for maximizing user participation through anchor node selection is entirely new.

Here is an outline of the paper. Section \ref{sec:Setup} introduces the notations and provides the necessary background in game theory. Section \ref{sec:Network Engagement games} formulates the network participation and resource-sharing problems using game theory, including a Nash Equilibrium analysis and simulation results on Erdős-Rényi random networks. Section \ref{sec:Stochastic Stability} presents the stochastic stability analysis for the network participation game. Section \ref{sec:InfluenceMaximization} details the game setup for maximizing user participation through anchor nodes. Finally, the paper concludes in Section \ref{sec:conclusion}.

\section{Setup} \label{sec:Setup}
We consider a social network of $n$ players represented by a graph $\mc{G}(\mc{V},\mc{E})$, where $\mc{V} = \{1,2,\ldots,n\}$ is the set of players and $\mc{E}$ is the set of edges. An edge $(i,j) \in \mc{E}$ indicates that players $i$ and $j$ can interact and mutually benefit. Let $\mc{N}_i$ be the neighborhood set of player $i$, i.e., 
\[
\mc{N}_i = \{j \in \mc{V} ~|~ (i,j) \in \mc{E} \}, 
\]
and $\mc{N}[i] = \mc{N}_i \cup \{i\}$ be the closed neighborhood of $i$. 
 Each player has a set of actions $\mc{A}_i$, and $\mc{A} = \mc{A}_1\times \mc{A}_2\cdots\times \mc{A}_n$ is the set of joint action profiles. Every element $\sigma$ in $\mc{A}$ is an $n$-tuple $(\sigma_1,\sigma_2,\ldots,\sigma_n) = (\sigma_i,\sigma_{-i})$, where $\sigma_i \in \mc{A}_i$ represents player $i$'s action and $\sigma_{-i}$ is an $(n-1)$-tuple that represents the actions of all the players other than $i$. Similarly, for any set $\mc{S} \subset \mc{V}$, $\sigma = (\sigma_{\mc{S}}, \sigma_{\mc{V} \bs \mc{S}})$ is a decomposition of action profile between players in the sets $\mc{S}$ and $\mc{V} \bs \mc{S}$, respectively. 

Players' preferences over the set of joint action profiles $\mc{A}$ are defined by their utility functions, $U_i~: \mc{A} \rightarrow \field{R}$. The set of best responses of player $i$ to  $\sigma_{-i}$ is 
\[
B_i(\sigma_{-i}) = \{\sigma^*_i \in \mc{A}_i ~|~ U_i(\sigma^*_i,\sigma_{-i}) \geq U_i(\sigma_i,\sigma_{-i}) ~\forall ~ \sigma_i \in \mc{A}_i\}.
\]
Let $d_H(\sigma,\sigma') = |\{j \in \mc{V} ~|~ \sigma_j \neq \sigma_j'\}|$ be the number of players with different actions in profiles $\sigma$ and $\sigma'$. A sequence of joint action profiles $\mc{P} = (\sigma^0,\sigma^1,\cdots,\sigma^{l-1})$ is a path if $\sigma^p \in \mc{A}$ for all $p \in \{0,1,\cdots,l-1\}$ and $d_H(\sigma^p,\sigma^{p+1}) = 1$ for each consecutive pair of profiles.  We define 
\[
\mc{P}_{\te{ind}}(\sigma^p,\sigma^{p+1}) = \{i \in \mc{V}~|~\sigma^p_{i} \neq \sigma^{p+1}_i\},
\]
i.e., $\mc{P}_{\te{ind}}(\sigma^p,\sigma^{p+1})$ is the index of the player updating their action from $\sigma^p$ to $\sigma^{p+1}$. The path $\mc{P}$ is a best response path if $\sigma^{p+1}_i$ belongs to $B_i(\sigma^p_{-i})$, where $i = \mc{P}_{\te{ind}}(\sigma^p,\sigma^{p+1}) $. 
Thus, at each step of the best response path, only one player's action is updated, and the updated action is the best response to the actions of other players.

In the theory of learning in games, various decision rules have been proposed and analyzed for modeling players' decision strategies (see, for instance, \cite{young1993evolution}, \cite{kandori_1993_learning}, \cite{Kandori1993}, \cite{ellison_2000_basins}, and \cite{newton2018evolutionary}). In this work, we adopt Log-Linear Learning (LLL), which models player decisions using noisy best response  (see \cite{Blume93}, \cite{Marden12}, and \cite{alos10} for details on LLL). Let $\sigma(t-1) = (\sigma_i(t-1),\sigma_{-i}(t-1))$ be the action profile at time $t-1$. Then, the steps involved in LLL at decision time $t$ are as follows:
\begin{itemize}
    \item A player, say player $i$, is randomly selected from the set of players $\mc{V}$.  
    \item The remaining players repeat their actions from iteration $t-1$, i.e., $\sigma_{-i}(t) = \sigma_{-i}(t-1)$.
    \item Player $i$ selects an action $\sigma_i' \in \mc{A}_i$ with the probability 
\begin{equation}\label{eq:pi_LLL}
\begin{aligned}
p^{\text{LLL}}_i(\sigma_i',\sigma_{-i}(t)) &= \frac{e^{-\frac{1}{T}[U_i(\sigma_i^*,\sigma_{-i}(t)) - U_i(\sigma_i',\sigma_{-i}(t)) ]}}{Z_i(\sigma_{-i}(t))}, \\
Z_i(\sigma_{-i}(t)) &= \sum\limits_{\bar{\sigma}_i \in \mc{A}_i} e^{-\frac{1}{T}[U_i(\sigma_i^*, {\sigma}_{-i}(t)) - U_i(\bar{\sigma}_i,\sigma_{-i}(t)) ]},
\end{aligned}
\end{equation}
where $Z_i(\sigma_{-i}(t))$ is the normalizing factor and $\sigma^*_i$ is an action from the best response set $B_i(\sigma_{-i}(t))$. 
\end{itemize}

The parameter $T$ in (\ref{eq:pi_LLL}), commonly referred to as the temperature or noise parameter, controls the level of exploration in decision-making. When $T = 0$, players do not explore and always select the best response to the actions of others. Increasing $T$ enhances exploratory behavior, allowing players to consider and select suboptimal actions. This exploration is crucial in complex optimization problems to avoid being trapped in local minima or maxima. As $T \rightarrow \infty$, players’ decisions approach a uniform distribution over their action sets, effectively resulting in random choices. Therefore, selecting an appropriate value for 
$T$ is essential to strike a balance between exploration and exploitation.

Since the decision of the randomly selected player at time $t$ only depends on the actions of others at time $t-1$, LLL induces a Markov chain over the set of joint action profiles $\mc{A}$. The induced Markov chain is ergodic and reversible with a unique stationary distribution $\mu^{\te{LLL}}_T$. An action profile $\sigma$ is \emph{stochastically stable} if and only if $\lim_{T \rightarrow 0} \mu_T^{\te{LLL}}(\sigma)>0$. The action profiles that are not stochastically stable have vanishingly small probabilities in the limiting case of $T \rightarrow 0$. Thus, stochastically stable profiles are good predictions for the long-run behavior of the system. 

A game is an \emph{exact potential game} if there exists a function $\Phi: \mc{A}\rightarrow \field{R}$ such that for any player $i \in \mc{V}$,
\begin{equation}\label{eq:PotentiaGame}
U_i(\sigma_i,\sigma_{-i}) - U_i(\sigma'_i,\sigma_{-i}) =\Phi(\sigma_i,\sigma_{-i}) - \Phi(\sigma'_i,\sigma_{-i}). 
\end{equation}
If this condition is satisfied, $\Phi$ is a potential function of the game. Thus, in a potential game, unilateral changes in any agent's utility quantitatively equals changes in a
potential function. For potential games under LLL, the set of stochastically stable profiles is equal to the set of action profiles that maximize a potential function of the game (\cite{Blume93} and \cite{alos10}). 

In \cite{voorneveld2000best}, a generalization of the potential game called the best-response potential game was presented. A game is a \emph{best-response potential game} if there exists a potential function $\Phi: \mc{A} \rightarrow \field{R}$, such that
\begin{equation}\label{eq:BR_Potential}
    \argmax_{\sigma_i \in \mc{A}_i } U_i(\sigma_i,\sigma_{-i}) = \argmax_{\sigma_i \in \mc{A}_i } \Phi(\sigma_i,\sigma_{-i}).
\end{equation}
Best-response potential games satisfy the property that a game with finite players and finite action sets is a best-response potential game if and only if any best-response path does not contain a cycle (Thm 3.2 of \cite{voorneveld2000best}). Moreover, according to Thm. 2 in \cite{alos10}, in a best-response potential games, if all the players adhere to LLL, then the set of stochastically stable action profiles is contained in the set of Nash equilibria. 
\section{Network Engagement Games} \label{sec:Network Engagement games}

We present a game-theoretic formulation of network engagement, where individuals must choose to either participate in a collaborative activity within their social network or refrain from doing so.
%
In particular, we consider the following two network engagement scenarios:
\begin{itemize}
    \item \emph{Network Participation Problem}: In this scenario, players decide whether to participate in a network activity. Their decision is influenced by the participation of their friends (or neighbors in a network graph). Specifically, a player will join the activity if a minimum number of their friends are already participating. 
    \item \emph{Resource-Sharing Problem}: In this case, players need access to resources that they may not possess entirely. Each player has a subset of resources, and cooperation is required for them to access resources held by others in the network. However, access is reciprocal: for a player to gain access to their neighbors' resources, they must also agree to share their own resources.
\end{itemize}

A key research question in network engagement games is how players decide to join or leave a cooperative activity or community. The threshold-based model, such as that proposed by \cite{granovetter1978threshold}, is prominent among diffusion-based models, where a user’s benefit from participation is a monotone function of the number of friends already engaged. However, empirical studies of online social networks reveal that diffusion alone does not fully account for network engagement evolution; players' individual behavioral characteristics also significantly influence their decisions. We address this discrepancy between the theoretical and empirical models by formulating network engagement as a non-cooperative game. In particular, we model players as boundedly rational decision-makers, allowing us to capture their unique decision behaviors, which correspond to non-diffusion network evolution.

\subsection{Network Participation Game }
To formulate the network participation problem as a game, we have a finite set of players $\mc{V} = \{1,2,\ldots,n\}$. Players has two actions in their action sets $A_i = \{1,0\}$, where the actions $1$ and $0$ represent participating or not participating in the activity. Given an action profile $\sigma$ in $\mc{A}$, we define 
\[
V_{\sigma} = \{j \in \mc{V} ~|~ \sigma_j = 1\}, 
\]
as the set of players participating in the network in action profile $\sigma$. We represent an action profile in which $p$ players participate as  $\sigma^{(p)}$. From a player's perspective, we define
\[
\mc{N}^p_i(\sigma) = \{j \in \mc{N}_i~|~\sigma_j=1\}.
\]
 where $\mc{N}^p_i(\sigma)$ is the set of neighbors of $i$ that are already participating in the network. Additionally, we use $\mc{N}^p_i[\sigma] =\mc{N}^p_i(\sigma) \cup \{i\}$ to represent the closed neighborhood of the participating players of  $i$. 
 
The first step in formulating the network participation problem as a non-cooperative game is to \emph{propose a utility function that aligns with empirical evidence and experimental observations of players' participation decisions while also enabling a comprehensive analysis of the game.}  The utility function of a player that we propose in this work is 
\begin{equation}\label{eq:Util-k}
U^{\te{NPG}}_i(\sigma_i,\sigma_{-i}) = 
\left\lbrace
\begin{array}{ll}
 \frac{\left(|\mc{N}^p_i(\sigma)| - k\right)}{|\mc{N}_i|}    &  \sigma_i = 1 \text{ and }|\mc{N}^p_i(\sigma)| < k, \\
 \alpha    & \sigma_i = 1 \text{ and }|\mc{N}^p_i(\sigma)| \geq k, \\
 0 & \sigma_i = 0.
\end{array}
\right.
\end{equation}
In the above expression, $\te{NPG}$ stands for Network Participation Game and $\alpha$ is a positive constant. Thus, player $i$ receives zero utility for not participating in the cooperative activity. The player's utility is negative for participating in the network if $|\mc{N}^p_i(\sigma)| <k$ and equals $\alpha$ if the number of participating neighbors of  $i$, $|\mc{N}^p_i(\sigma)|$, is greater than the threshold value $k$. This utility function captures the idea that players are motivated to participate if their neighbors' participation exceeds a certain threshold $k$, while participation without sufficient neighbor support is penalized. The utility function in (\ref{eq:Util-k}) aligns with the empirical findings in \cite{backstrom2006group} and \cite{wu2013arrival} that the players are more likely to join as the number of active neighbors increases, but the marginal impact of additional active neighbors diminishes. In the proposed utility function, a player's utility initially increases linearly with the number of participating neighbors. However, once the threshold $k$ is reached, the player receives a fixed benefit of $\alpha$, which remains constant regardless of further participation from additional neighbors.

The utility function is also normalized by $|\mc{N}_i|$, which effectively updates the threshold value or cost for participation from $k$ to $k/|\mc{N}_i|$. An impact of this normalization is that more influential players with large neighborhoods will be more willing to explore because their probability of selecting suboptimal actions will be higher than those with a few neighbors \cite{malliaros2013stay} and \cite{wu2013arrival}. Thus, well-connected players in the network will likely join earlier than players with few neighbors. 

\subsection{Network Resource Sharing Game}\label{sec:NSG}
The second setup of network engagement we consider is a generalization of the network participation game as presented in \cite{abbas2019graph}. In this generalized setup, each player possesses a set of personal resources, which can include physical sensors, information, or expertise for performing specific tasks. Players can share their resources with immediate neighbors in the network. Let $\mc{R} = \{0,1,\ldots,r-1\}$ denote the set of available resources in the network. Each player is assigned a subset of $s$ resources, where $s \leq r$. This setup can be represented as a labeled graph in which the label of a node corresponds to the set of resources allocated to that player.

Given an action profile $\sigma = (\sigma_i, \sigma_{-i})$ in $\mc{A}$,  we define
\begin{equation}\label{eq:L}
    L^p_i(\sigma) =  \bigcup_{j \in \mc{N}^p_i[\sigma]} l(j).
\end{equation}
 Here, $\mc{N}^p_i[\sigma]$ is the closed neighborhood of $i$, which includes $i$ and all the participating neighbors of $i$ in profile $\sigma_{-i}$. Thus, $L^p_i(\sigma)$ is the set of resources that $i$ can access either directly or through its participating neighbors. Moreover, we define 
 $$L_i = \cup_{j \in \mc{N}[i]} l(j),$$ 
 which is the set of resources available in the closed neighborhood of $i$. 
 
To propose a utility function for the resource-sharing setup, we use the principle of reciprocity as a guiding concept (see \cite{nowak1998evolution} and \cite{surma2016social}. Reciprocity, which has been widely studied in social psychology, suggests that individuals contribute based on what they receive from others. In the context of resource-sharing, this means that players are more inclined to share their resources if they can access the resources they need from their neighbors. We formalize this principle by proposing a utility function for the players. The utility function reflects the idea that players are willing to share resources if doing so allows them to obtain the resources they are missing from their network of neighbors. The specific utility function is defined as follows:
\begin{equation}\label{eq:U_rs}
U^{\te{NSG}}_i(\sigma_i,\sigma_{-i}) = 
\left\lbrace
\begin{array}{cc}
 \frac{\left(|L^p_i(\sigma)| - r\right)}{|L_i|}    &  \sigma_i = 1 \text{ and }|L^p_i(\sigma)| < r, \\
 \alpha    & \sigma_i = 1 \text{ and }|L^p_i(\sigma)| = r, \\
 0 & \sigma_i = 0,
\end{array}
\right.
\end{equation}
where $\te{NSG}$ stands for Network Sharing Game and $\alpha$ is a positive constant. Player $i$ receives zero utility for not participating in the sharing activity. If the player decides to participate, the resulting utility, as defined in (\ref{eq:U_rs}), is negative if the number of resources that $i$ can access in the closed neighborhood of participating players, $\mc{N}^p_i(\sigma)$, is less than $r$. However, if $i$ has access to all the $r$ resources, then $U_i$ equals a positive utility of $\alpha$. In this scenario, the utility is adjusted for the size of the set $|L_i|$. This means that players whose local neighborhoods do not have access to the complete set of resources will be less inclined to participate compared to those whose neighborhoods have all the resources.


\subsection{Analysis}
We begin by providing a qualitative description of the Nash equilibria to analyze the equilibrium characteristics of the network participation and resource-sharing games. This involves characterizing the conditions under which players' strategies form a stable equilibrium where no player has an incentive to unilaterally deviate.
In the next section, we present stochastic stability analysis for the network participation game for various network configurations, including line, ring, wheel, and two-dimensional grid topologies. The goal is to understand how the equilibrium behavior persists under stochastic perturbations and to identify which equilibrium configurations represent stable outcomes.

\subsubsection{Network Participation Game}\label{Sec:NPG}
We first analyze the network participation game for which the player utilities are defined in (\ref{eq:Util-k}), and players employ a combination of explorations and exploitation in Log-Linear Learning (LLL) as their decision strategy. We will use Nash equilibrium and stochastic stability as our solution concept to analyze network formation behavior for this setup. 
\begin{prop}\label{prop:NE_k}
    For the network participation game with utility function defined in (\ref{eq:Util-k}), an action profile $\sigma^*$ is a Nash equilibrium if either of the following conditions is satisfied. 
    \begin{enumerate}
        \item $\sigma_i^* = 0$ for all $i \in \mc{V}$. 
        \item For all $i$ in $V_{\sigma^*}$, $d(i,V_{\sigma^*}) \geq k$ and for all $j$ in $V \bs V_{\sigma^*}$, $d(j,V_{\sigma^*}) <  k$.       
    \end{enumerate}
\end{prop}
\begin{proof}
To prove the first part of the proposition, suppose the joint action profile is $\sigma^{(0)} = (0,0\ldots,0)$ resulting in zero utility for all the players. At each iteration, a random player, say $i$, is selected to action revision under LLL. Then,
\[
U^{\te{NPG}}_i(0,\sigma_{-i}) = 0 \text{ and }U^{\te{NPG}}_i(1,\sigma_{-i}) = -k/|\mc{N}_i|.
\]
Thus, each player prefers no participation over participation if the current joint profile is $(0,\ldots,0)$. 

For condition 2), the joint action profile $\sigma $ is represented as $\sigma = (\sigma_{V_{\sigma^*}},\sigma_{V \bs V_{\sigma^*}})$. For each player in $V_{\sigma^*}$, the condition $d(i,V_{\sigma^*}) \geq k$ implies that player $i$ has at least $k$ neighbors who are participating in the activity. Therefore, $U^{\te{NPG}}_i(1,\sigma^*_{-i}) > U^{\te{NPG}}_i(0,\sigma^*_{-i})$. The strict inequality is enforced by a positive value of $\alpha$ in the utility function when the number of participating neighbors is greater or equal to $k$.  Similarly, for all players $j$ in $V \bs V_{\sigma^*}$, the condition $d(j, V_{\sigma^*}) < k$ implies that the number of participating neighbors is less than $k$ and therefore $U^{\te{NPG}}_i(0,\sigma^*_{-i}) < U^{\te{NPG}}_i(1,\sigma^*_{-i})$.
\end{proof}

\begin{prop}\label{prop:acyclic_k}
    For the network participation game with utility function defined in (\ref{eq:Util-k}), all best response paths are acyclic. 
\end{prop}
\begin{proof}
    Suppose the above statement is not true and there exists a best response path $\mc{P} = (\sigma^1,\sigma^2,\cdots, \sigma^{l},\sigma^{l+1})$ that is a cycle, i.e., $\sigma^1 = \sigma^{l+1}$. For a path to be a cycle, every player that transitions from $0$ to $1$ must transition back to $1$ and every player that transitions from $1$ to $0$ must transition back from $0$ to $1$. Thus, we can only have a cycle with an even number of transitions, which implies that $l$ must be an even number with $ \sigma^{l+1} = \sigma^1$.

    Let $G$ be a global function on the set of joint action profiles and is defined as
    \[
        G(\sigma) = \sum_{i=1}^n  \hat{U}_i(\sigma_i,\sigma_{-i}) , 
    \]
    where 
    \begin{equation}
       \hat{U}_i(\sigma_i,\sigma_{-i}) = \sigma_i \left(|\mc{N}^p_i(\sigma)|  - k \right).
    \end{equation}
    If the path $\mc{P}$ is a cycle, then the following equality must hold
    \[
    \sum_{q = 1}^{l} \left[G(\sigma^{q+1})  - G(\sigma^q)\right] = 0. 
    \]
    Since every transition must be reversed in a cyclic path. Suppose
    \[
      \mc{P}(\sigma^{q_1},\sigma^{q_1+1}) = \mc{P}(\sigma^{q_2},\sigma^{q_2+1}) = i,
    \]
    i.e., player $i$ updates the actions in transitions $(\sigma^{q_1},\sigma^{q_1+1})$ and $(\sigma^{q_2},\sigma^{q_2+1})$. Without loss of generality, we assume that $i$ transitions from $0$ to $1$ in $\sigma^{q_1}$ to $\sigma^{q_1+1}$ transition and from $1$ to $0$ in some later transition $\sigma^{q_2}$ to $\sigma^{q_2}$, where $1< q_1 < q_2<l $. Then, 
    \begin{align*}
        &G(\sigma^{q_1+1})-G(\sigma^{q_1}) = [\hat{U}_i(1,\sigma^{q_1}_{-i}) - \hat{U}_i(0,\sigma^{q_1}_{-i}] + \\
       & \sum_{j\in \mc{N}^p_i(\sigma^{q_1})} \left[\hat{U}_j(\sigma^{q_1}_j,1,\sigma^{q_1}_{-\{i,j\}}) -~  \hat{U}_j(\sigma^{q_1}_j,0,\sigma^{q_1}_{-\{i,j\}}) \right]
+ \\      &\sum_{j\notin \mc{N}^p_i[\sigma^{q_1}]} \left[\hat{U}_j(\sigma^{q_1}) - \hat{U}_j(\sigma^{q_1}) \right].
    \end{align*}
    Here $\hat{U}_j(\sigma^{q_1}_j,1,\sigma^{q_1}_{-\{i,j\}})$ is the updated utility of player $j$ when $\sigma^{q_1}_j$ is the action of $j$, $\sigma^{q_1}_i = 1$ is the action of player $i$ and $\sigma^{q_1}_{-\{i,j\}}$ is the joint action profile of actions of players other than $i$ and $j$. 
    
    The last term in the above equation is equal to zero since the utility of the players that are not in the participating neighborhood of player $i$ remains the same from $\sigma^{q_1}$ to $\sigma^{q_1+1}$. Player $i$ transitions from $0$ to $1$ as a best response if $|\mc{N}^p_i(\sigma^{q_1})|\geq k$. Thus, 
    \[
   [\hat{U}_i(1,\sigma^{q_1}_{-i}) - \hat{U}_i(0,\sigma^{q_1}_{-i}) ] = |\mc{N}^p_i(\sigma^{q_1})|-k.     
    \]
    The second term in the expression for $G(\sigma^{q_1}+1)-G(\sigma^{q_1})$ is the impact of player $i$'s participation on the utility of his participating neighbors and is equal to $\mc{N}^p_i(\sigma^{q_1})$ since the utility of each participating neighbor of $i$ is increased by a factor of $1/\mc{N}_j$. Thus, 
    \[
    G(\sigma^{q_1+1})-G(\sigma^{q_1}) = 2\mc{N}^p_i(\sigma^{q_1}) - k.
    \]

    When player $i$ transitions from $1$ to $0$ at some later step of the cycle, say $\sigma^{q_2} = (1,\sigma^{q_2}_{-i})$ and $\sigma^{q_2+1} = (0,\sigma^{q_2}_{-i})$, then 
    \[
    G(\sigma^{q_2+1})-G(\sigma^{q_2}) = -(2\mc{N}^p_i(\sigma^{q_2}) - k).
    \]
    Since a transition from $1$ to $0$ is the best response only if the number of participating neighbors is less than $k$, we get 
    $\mc{N}^p(\sigma^{q_2}) < k \leq \mc{N}^p(\sigma^{q_1})$. Therefore, 
    \[
     [G(\sigma^{q_1+1})-G(\sigma^{q_1})] + [G(\sigma^{q_2+1})-G(\sigma^{q_2})] > 0.
    \]
    Thus, given the best response cycle, we can divide the entire path into a pair of transitions of individual players from $1$ to $0$ and then from $0$ to $1$. The corresponding difference in the global function $G(\cdot)$ for these transition pairs is always greater than zero. Therefore, \[
    \sum_{p = 1}^{l-1} \left[G(\sigma^{p+1})  - G(\sigma^p)\right] > 0 
    \]
    for any best response cycle, which is a contradiction. Thus, the best response path cannot be a cycle for our network participation game with the utility function defined in (2). 
    \end{proof}
\begin{prop}
    Consider the network participation game with utility function defined in (\ref{eq:Util-k}). If all the players adhere to Log-Linear Learning (LLL), then the set of stochastically stable action profiles belongs to the set of Nash equilibria. 
\end{prop}
\begin{proof}
    The proof is a direct consequence of Props. \ref{prop:NE_k} and \ref{prop:acyclic_k}.  Prop. \ref{prop:NE_k} implies that all the Nash equilibria of the game are strict, i.e. if $\sigma^* \in \mc{A}$ is a Nash equilibrium then $U^{\te{NPG}}_i(\sigma^*_i,\sigma^*_{-i})$ is strictly greater than $U^{\te{NPG}}_i(\sigma_i,\sigma^*_{-i})$ for any $\sigma_i \in A_i$. Prop \ref{prop:acyclic_k} establishes that the best response paths are acyclic. For a finite number of players with a finite number of actions, $|\mc{A}|$ is finite, and therefore, every best response path should have a finite length. 
    Thus, every best response path has to terminate, and it cannot terminate to any profile other than a Nash equilibrium. The Markov chain induced by LLL follows a best response path with high probability and Nash equilibria are the only absorbing states of the Markov chain if $T = 0$. Thus, in the limiting case of $T \rightarrow 0$, the stochastically stable profiles for which $\mu_T^{\te{LLL}} >0 $ will belong to the set of Nash equilibria, which concludes the proof. 
\end{proof}

\subsubsection{Network Resource Sharing Game}\label{Sec:NSG}
Next, we characterize the set of Nash equilibria for the resource-sharing game setup. 

\begin{prop}\label{prop:NE_rs}
    For the network sharing game with utility function defined in (\ref{eq:U_rs}), an action profile $\sigma^*$ is a Nash equilibrium if either of the following conditions is satisfied. 
    \begin{enumerate}
        \item $\sigma_i^* = 0$ for all $i \in V$. 
        \item For all $i$ in $V_{\sigma^*}$, $|L_i(\sigma^*)| = r$ and for all $j$ in $V \bs V_{\sigma^*}$, $|L_i(\sigma^*)|  <  r$, where $L_i(\sigma^*)$ is the set of resources that player $i$ can access in the closed neighborhood of participating players.     
    \end{enumerate}
\end{prop}  

\begin{proof}
The proof follows the same reasoning as Prop. \ref{prop:acyclic_k}. We refer the reader to our earlier work \cite{luqman2024dynamic} for the details. 
\end{proof}

\subsection{Simulation}\label{sec:Simulation_local}
We validated the effectiveness of our proposed models through simulations, using Erdős-Rényi (ER) graphs to represent the underlying social network. We compared the resulting participation and resource-sharing networks with the $k$-core and $(r, s)$-core structures from the literature. The simulations showed that our framework consistently converges to these core structures, even when starting from zero participation or random initial states. This independence from initial condition distinguish our approach from prior methods, which often assume full initial participation, enabling a more comprehensive analysis of network evolution and stability.

In an Erdős-Rényi (ER) graph, edges form between any pair of nodes with probability $p$. For our experiments, we used ER graphs with $n = 1000$ players and a link probability of $p = 0.01$, giving each player an average of ten neighbors. We generated 100 random ER graphs and verified our results across all of them, then computed average performance metrics to ensure the robustness of our findings.

\begin{figure*}[t]
    \centering
    \subfigure[Participation Network: $k$-core size]{
        \includegraphics[scale=0.48]{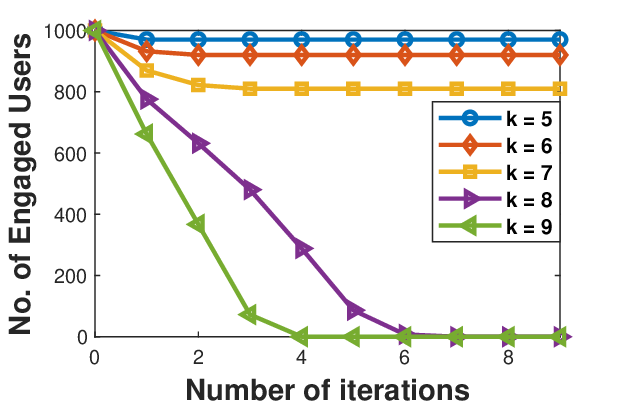}
        \label{fig:kcore}
    }
    \hspace{0.1in}
    \subfigure[Participation Network using LLL (Initial condition: $\sigma^{(0)}$)]{
        \includegraphics[scale=0.48]{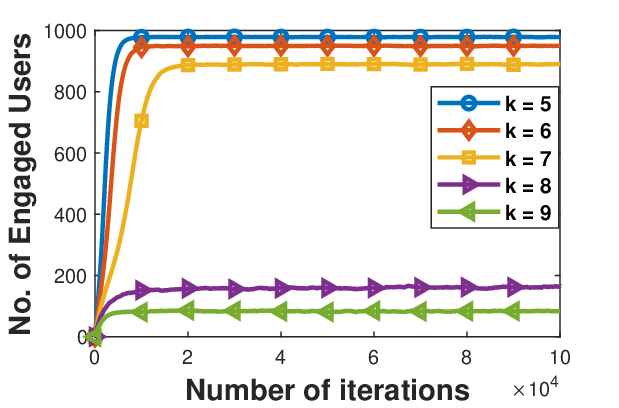} 
        \label{fig:LLL_kcore}
    }    
    \hspace{0.1in}
    \subfigure[Participation Network  using LLL: Random initial conditions]{
        \includegraphics[scale=0.48]{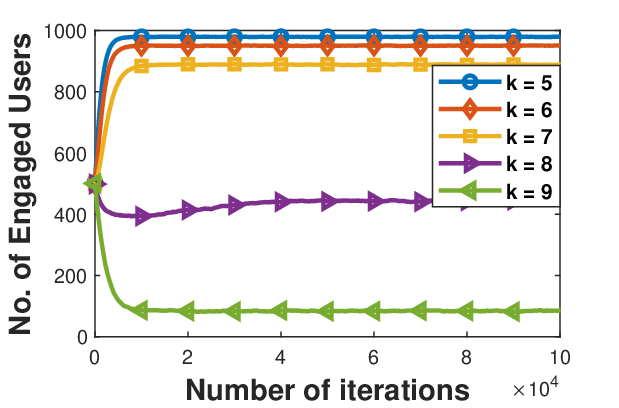} 
        \label{fig:LLL_kcore_rand}
    }    \\
    \subfigure[Resource Sharing Network: $(r,s)$-core size]{
        \includegraphics[scale=0.48]{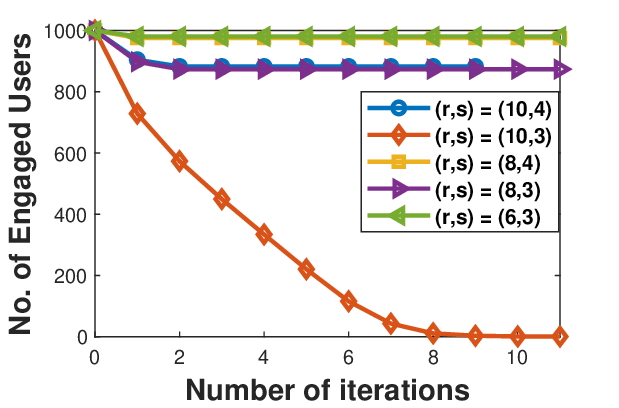} 
        \label{fig:rscore}
    }
    \hspace{0.1in}
    \subfigure[Resource Sharing Newtork using LLL (Initial condition: $\sigma^{(0)}$) ]{
        \includegraphics[scale=0.48]{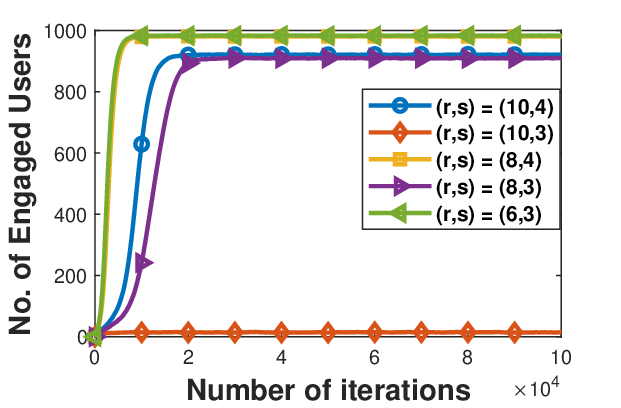} 
        \label{fig:rscore_LLL}
    }
    \hspace{0.1in}
    \subfigure[Sharing Network using LLL (Random initial conditions)]{
        \includegraphics[scale=0.48]{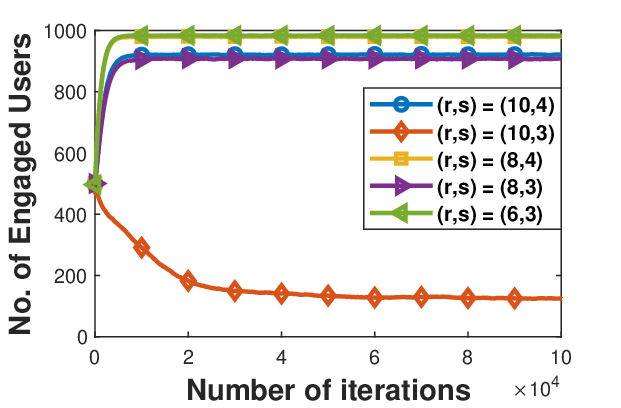} 
        \label{fig:rscore_rand}
    }
    \caption{Simulation results for the network participation and resource sharing Games}
    \label{fig:combined_results}
    \end{figure*}
\subsubsection{Network Participation Game}For the network participation game, we set the Log-Linear Learning (LLL) noise parameter to $T = 0.3$ and the utility parameter $\alpha = 1$, as defined in equation (\ref{eq:Util-k}). We ran simulations for threshold values $k = \{5, 6, 7, 8, 9\}$, representing the minimum number of neighbors required for participation, with results shown in Fig. \ref{fig:combined_results}. Fig. \ref{fig:kcore} illustrates the outcomes of the standard $k$-core selection algorithm \cite{bhawalkar2015preventing}, where players not meeting the $k$-neighbor condition are iteratively removed, starting from full participation. The horizontal axis represents the number of iterations, and the vertical axis shows the number of engaged players after each iteration, which converges to the number of players in the $k$-core. The results are averaged over 100 random ER networks.  As shown in the figure, the $k$-core is empty for $k = 8$ and $k = 9$, while a substantial fraction of players are retained for $k < 8$.

The plots in Figs. \ref{fig:LLL_kcore} and \ref{fig:LLL_kcore_rand} show the results of our proposed network participation approach with zero and random initial conditions, respectively. The horizontal axis represents the number of iterations of the Log-Linear Learning (LLL) algorithm, while the vertical axis corresponds to the number of players participating in the activity. Participation steadily increases under LLL, stabilizing after about 20,000 iterations. The results demonstrate a clear relationship between the $k$-core size in Fig. \ref{fig:kcore} and the final network size across different threshold values of $k$, highlighting that our approach effectively promotes network participation, even from minimal or random initial states.

We compared each player's actions after $k$-core formation via the iterative withdrawal algorithm to their steady-state actions under LLL. Due to the exploratory behavior from the noise parameter $T$, players occasionally change their actions, but such events are rare. To assess steady-state behavior, we analyzed player actions over the final 30,000 LLL iterations, during which the network state was stable (as shown in Figs. \ref{fig:LLL_kcore} and \ref{fig:LLL_kcore_rand}). We applied the criterion that if a player participated in more than 95\% of these iterations, they were considered a network member; otherwise, their participation was considered sporadic and exploratory.

\begin{table}[t]
\begin{center}
\begin{tabular}{ |c|c | c|c|c|c|}
\hline
Network size $/$ $k$ & 5 & 6 & 7 & 8 & 9  \\ 
\hline
Cascading Withdrawal ($\sigma{(0)}= \sigma^{(n)})$ & 973 &919&820 &0&0 \\
 \hline
LLL ($\sigma(0) = \sigma^{(0)})$ & 971 &915&797 &0&0 \\
\hline
LLL ($\sigma(0) = $ Random) &970 &915&797 &279& 0 \\
\hline
\end{tabular}
\end{center}
\caption{ Performance comparison for Network Participation Game}
\label{table:kcore}
\vspace{-0.3in}
\end{table}

The results of this analysis are summarized in Table \ref{table:kcore} for two scenarios: no initial participation ($\sigma^{(0)}$) and random initial conditions. For $k \in \{5,6,7,9\}$, the expected number of nodes in the $k$-core, computed via the \emph{cascading eithdrawal algorithm}, is compared with the number of participating nodes at steady state under LLL. The difference is negligible in all cases except for $k=8$ under random initial conditions. These results suggest that in the network participation game, where players maximize utility using LLL, as in equation (\ref{eq:Util-k}), the $k$-core of the underlying social network is a strong predictor of long-term participation from any initial condition. Additionally, our framework effectively captures both diffusion-based and non-diffusion-based network evolution, offering a comprehensive model of participation dynamics over time.

\subsubsection{Network Resource Sharing Game}
We conducted a similar analysis for the Network Sharing Game, and the results are shown in Fig. \ref{fig:combined_results}. Using an ER graph with $n = 1000$ and $p = 0.01$, each player had an average of ten neighbors. We tested different $(r,s)$ parameter pairs: $(10,4)$, $(10,3)$, $(8,4)$, $(8,3)$, and $(6,3)$. For the utility function in (\ref{eq:U_rs}), we set the noise parameter $T = 0.15$ and $\alpha = 1$. The results were averaged over one hundred randomly generated ER networks. In Fig. \ref{fig:rscore}, we show the expected number of players included in the $(r,s)$-core of the ER networks. The $(r,s)$-core was empty only for the case $(r,s) = (10,3)$. For all other scenarios, more than $95\%$ of players were included in the core. Similarly, Fig. \ref{fig:rscore_LLL} presents the results for our proposed scheme in the network sharing game, where the vertical-axis indicates the number of players sharing resources according to the LLL decision strategy. The results closely align with those from the $(r,s)$-core analysis, demonstrating consistency in network behavior under our framework.

We compared the network performance under the $(r,s)$-core model with the size of the sharing network using our proposed scheme. Specifically, we analyzed the actions of all players in both the $(r,s)$-core and the steady state of our scheme, where we consider the steady state for the final thirty thousand iterations. The results, shown in Table \ref{table:rscore}, reveal that the discrepancies between the two approaches are negligible across all scenarios. This close agreement confirms that the players' actions under LLL align well with the $(r,s)$-core setup, validating the effectiveness of our proposed approach for the network sharing game as detailed in Section \ref{sec:NSG}.

\begin{table}[t]
\begin{center}
\begin{tabular}{ |c|c | c|c|c|c|}
\hline
Network size $/$ $(r,s)$ & (10,4) & (10,3) & (8,4) & (8,3) & (6,3)  \\ 
\hline
Cascading Withdrawal  & 885 &0&983 &870&980 \\
 \hline
LLL ($\sigma(0) = \sigma^{(0)})$ & 866 &2&972 &850&969 \\
\hline
LLL ($\sigma(0) = $ Random) &866 &56&972 &849& 970 \\
\hline
\end{tabular}
\end{center}
\caption{ Performance comparison for Network Participation Game}
\label{table:rscore}
\vspace{-0.3in}
\end{table}

\begin{figure*}[t]\label{fig:RingandWheel}
    \centering
    \subfigure[Ring network ]{
        \includegraphics[scale=0.45]{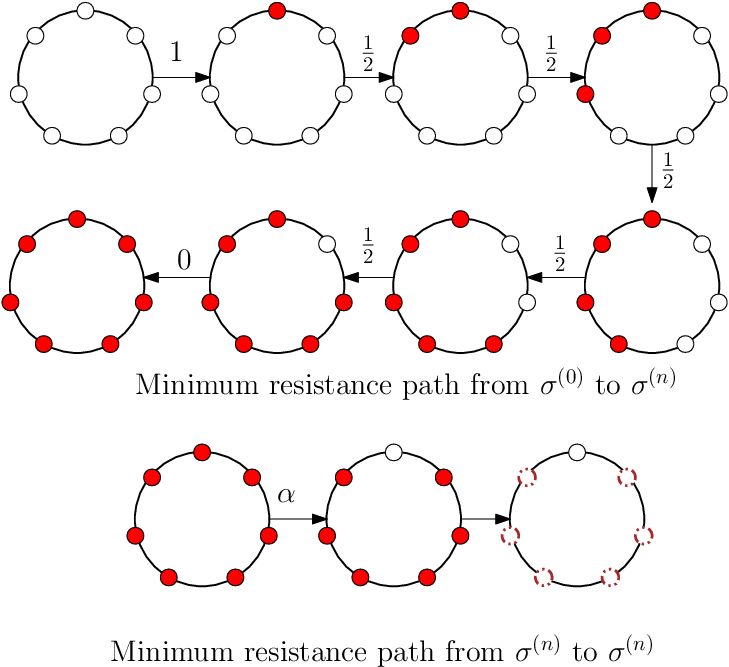}
        \label{fig:ring}
    }
    \hspace{0.2in}
    \subfigure[Wheel network]{
        \includegraphics[scale=0.55]{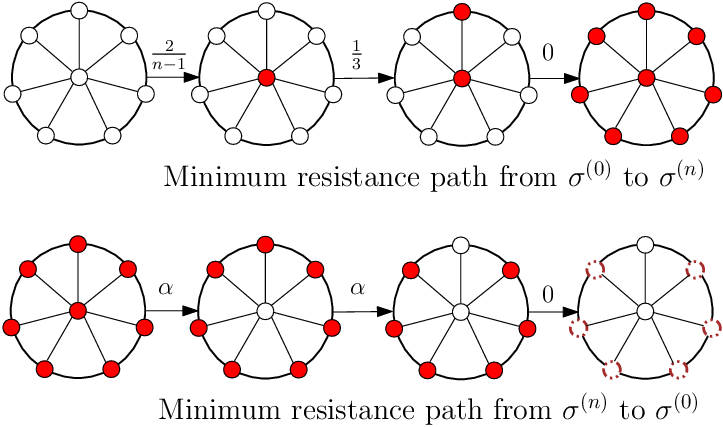} 
        \label{fig:wheel}
    }    
    \caption{Minimum resistance paths between Nash equilibria for Ring and Wheel networks. The numbers on the edges represent transition resistance.}
\end{figure*} 
\section{Stochastic Stability Analysis of Network Participation Game}\label{sec:Stochastic Stability}
In the previous section, we characterized the set of Nash equilibria for the network participation and resource-sharing games, demonstrating that stochastically stable profiles are contained within this set. In this section, we further refine our understanding of the long-run behavior of the network participation game by identifying the stochastically stable network configurations for specific classes of social networks. In particular, we focus on scenarios where the underlying social network is represented by one of the following graph types: line graph, ring graph, wheel graph, or two-dimensional grid. For this analysis, we restrict our study to the case where the participation threshold is set to $k = 2$ for all players. Our approach to identifying stochastically stable profiles is based on Radius-CoRadius analysis, as outlined in (\ref{eq:RdCR}). 

\subsection{Radius CoRadius Analysis}\label{sec:RdCR}
Our stochastic stability analysis is based on the Radius-CoRadius ($Rd-CR$) criteria presented in \cite{alos10}. This analysis offers a sufficient condition for identifying stochastically stable action profiles in Log-Linear Learning dynamics by comparing two parameters, namely Radius ($Rd$) and CoRadius ($CR$) of an action profile. This sufficient condition is derived based on the results from resistance tree analysis of the stationary distribution of ergodic Markov chain induced by LLL over the set of joint action profiles. Given an action profile pair $\sigma = (\sigma_i,\sigma_{-i})$ and $\sigma' = (\sigma'_i,\sigma_{-i})$ in which only one player has a different action, we define transition resistance from $\sigma$ to $\sigma'$ as
\[
R(\sigma,\sigma') = \max_{\hat{\sigma}_i \in \mc{A}_i} U^{\te{NPG}}_i(\hat{\sigma}_i,\sigma_{-i}) - U^{\te{NPG}}_i(\sigma'_i,\sigma_{-i}). 
\]
The resistance of a path $\mc{P} = (\sigma^0,\sigma^1,\cdots,\sigma^{l-1})$ is 
\[
R^p(\mc{P}) = \sum_{t = 0}^{l-2} R(\sigma^t,\sigma^{t+1}).
\]
 For any pair of profiles $\sigma$ and $\sigma'$ in $\mc{A}$, let $\Omega(\sigma,\sigma')$ be the set of all paths from $\sigma$ to $\sigma'$. 
Then, $R_{\te{min}}(\sigma,\sigma')$ be the minimum resistance  of all the paths in  $\Omega(\sigma,\sigma')$ . The \emph{basin of attraction} of an action profile $\sigma$, $\te{BoA}(\sigma)$, is the set of all profiles $\sigma' \in \mc{A}$ such that there exists a path of zero resistance from $\sigma'$ to $\sigma$.  The Radius of an action profile $\sigma$ is 
\begin{equation}\label{eq:Radius}
    Rd(\sigma) = \min \{R_{\te{min}}(\sigma,\sigma') ~\text{ for all } \sigma' \in \mc{A} \bs \te{BoA}(\sigma)\},
\end{equation}
and the CoRadius of $\sigma$ is
\begin{equation}\label{eq:CoRadius}
    CR(\sigma)  = \min\{R_{\te{min}}(\sigma,\sigma') \text{ for all } \sigma'
    \in \mc{A} \bs \te{BoA}(\sigma)\}.
\end{equation}
The Radius of a profile measures how easy it is to leave that profile, and the Coradius measures how difficult it is to reach it starting from any initial condition. 
\begin{definition} (\emph{$Rd-CR$ Criteria}) \cite{alos10}
    Given an action profile $\sigma \in \mc{A}$,  if
\begin{equation}\label{eq:RdCR}
Rd(\sigma) > CR(\sigma),
\end{equation}
then, $\sigma$ is stochastically stable. 
\end{definition}

Next, we identify stochastically stable profiles for the network participation game across different social network configurations using the $Rd-CR$ result outlined in equation (\ref{eq:RdCR}). To perform the $Rd-CR$ analysis, we first need to determine the minimum resistance paths between Nash equilibria. We denote an action profile where $p$ players participate as $\sigma^{(p)}$.\\
\emph{Assumption:} For the following analysis, we assume that the participation threshold is $k = 2$. 

\subsubsection{Line Graph} 
If the players are arranged in a line topology, the participation game will have a unique equilibrium in which no player participates. The reason is that the two end users only have one neighbor, and they will not participate. Their lack of participation will diffuse through the network through their neighbors.

   \begin{figure*}[t]
    \centering
    \subfigure{
        \includegraphics[scale=0.88]{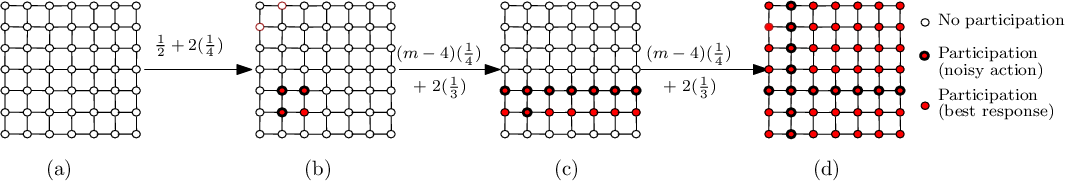}
            }
    \caption{A minimum resistance path from $\sigma^{(0)}$ to $\sigma^{(n)}$ for $CR(\sigma^{(n)}$. Empty nodes represent players who are not participating, red nodes with dark boundaries  represent players who are participating as noisy action, and red nodes with regular boundaries represent players who started to participate as their best action.}
    \label{fig:2DGrid}
\end{figure*} 

\begin{figure*}[t]
    \centering
    \subfigure{
        \includegraphics[scale=0.75]{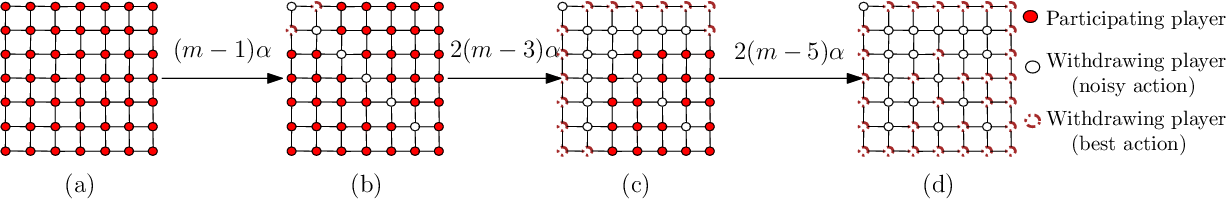}
            }
    \caption{A minimum resistance path from $\sigma^{(n)}$ to $\sigma^{(0)}$ for $CR(\sigma^{(0)})$. Solid nodes (red color) represent participating players, empty nodes with regular boundaries represent players who left the network as noisy action, and empty nodes with dashed/dotted lines represent players who left the network as their optimal action. }
    \label{fig:CR_2DGrid}
\end{figure*} 

\subsubsection{Ring Network} In the ring network shown in Fig. \ref{fig:ring}, each player has two neighbors. The game has two Nash equilibria, one in which all players participate, $\sigma^{(n)}$, and the other in which no player participates, $\sigma^{(0)}$. The minimum resistance paths between the two equilibria are presented in Fig. \ref{fig:ring}. The Radius of $\sigma^{(0)}$ and $\sigma^{(n)}$ are 
\begin{align*}
    Rd(\sigma^{(0)}) &= R_{\te{min}}(\sigma^{(0)},\sigma^{(n)}) = 1 + (n-2)/2 = n/2,\\
    Rd(\sigma^{(n)}) &= R_{\te{min}}(\sigma^{(n)},\sigma^{(0)}) = \alpha. 
\end{align*}
In the minimum resistance path from $\sigma^{(0)}$ to $\sigma^{(n)}$,  presented in Fig. \ref{fig:ring}, the resistance of the first step is 
\[
R(\sigma^{(0)},\sigma^{(1)}) = U^{\te{NPG}}_i(0,\sigma^{(0)}_{-i}) - U^{\te{NPG}}_i(0,\sigma^{(1)}_{-i}) = 1. 
\]
For the subsequent $n-2$ steps, the resistance is $1/2$, which leads to the above expression for $Rd(\sigma^{(0)})$. The minimum resistance path between $\sigma^{(n)}$ and $\sigma^{(0)}$ includes only one noisy action of resistance $\alpha$. 

Since we have only two equilibria, $CR(\sigma^{(n)}) = Rd(\sigma^{(0)})$ and $CR(\sigma^{(0)}) = Rd(\sigma^{(n)})$. Thus, from (\ref{eq:RdCR}),  if
\[
 \alpha > \alpha_{\te{th}} = n/2 \text{ then } Rd(\sigma^{(n)}) > CR(\sigma^{(0)}),
\]
then $\sigma^{(n)}$, the NE in which all the players participate, is stochastically stable

\subsubsection{Wheel Network} In the wheel network of $n$ players, a central player is added to the ring network connected with all the players. When the players are connected in the wheel topology, the participation game will again have two Nash Equilibria, i.e., either no player participates ($\sigma^{(0)}$)  or all the players participate ($\sigma^{(n)}$). The minimum resistance paths between the two equilibria are presented in Fig. \ref{fig:wheel}. The Radius of equilibria $\sigma^{(0)}$ and $\sigma^{(n)}$ are
\begin{align*}
    Rd(\sigma^{(0)}) &= R_{\te{min}}(\sigma^{(0)},\sigma^{(n)}) = \frac{1}{3}+ \frac{2}{n-1} ,\\
    Rd(\sigma^{(n)}) &= R_{\te{min}}(\sigma^{(n)},\sigma^{(0)}) = 2\alpha. 
\end{align*}
For two Nash Equilibria, $CR(\sigma^{(n)}) = Rd(\sigma^{(0)})$ and $CR(\sigma^{(0)}) = Rd(\sigma^{(n)})$. Therefore, if 
\begin{equation}\label{eq:RdCR_wheel}
    \alpha > \alpha_{\te{th}} ~=~ \frac{1}{n-1} + \frac{1}{6},
\end{equation}
then $\sigma^{(n)}$ is stochastically stable. Otherwise, $\sigma^{(0)}$ is stochastically stable. Here, $\alpha_{\te{th}}$ represents the threshold value for the benefit that a player must receive for the participation to be considered stochastically stable.

Notice the difference in the condition on $\alpha$ between ring and wheel networks due to the addition of the central player. In the ring network, the payoff increases with the number of players $n$, making $\sigma^{(0)}$ stochastically stable unless $\alpha$ is very large. In the wheel network, however, $\alpha$ depends on a constant factor and a term inversely related to $n$, with the central player's influence playing a pivotal role. Normalizing the utility by $|\mc{N}_i|$ further reduces the central player's resistance to selecting noisy actions, encouraging exploration as the number of neighbors increases.

\subsubsection{Two-dimensional Grid}  An $m\times m$ grid network has $4$ corner nodes with two neighbors, $4(m-2)$ boundary nodes with three neighbors and $(m-2)^2$ internal nodes with four neighbors. When players are arranged in a grid configuration and $k=2$, the participation game has multiple Nash Equilibria. All the configurations in Figs. \ref{fig:2DGrid} and \ref{fig:CR_2DGrid} are Nash equilibria for the $7\times 7$ grid network.  


Using the $Rd-CR$ analysis, we derive the condition on $\alpha$ that will result in $\sigma^{(n)}$ as the stochastically stable profile. Here $n = m^2$ is the total number of nodes/players in the network. Starting from an $m\times m $ grid of participating players, as shown in Fig. \ref{fig:2DGrid}(d) for $m = 7$, six players (top six nodes in the second column from the left) must leave the network and play noisy actions to reach the next NE. Note that this configuration is one of the many possibilities of leaving $\sigma^{(n)}$. However, all the possibilities involve six noisy actions if $m = 7$. For an $m\times m$ grid, this observation generalizes to $(m-1)$ noisy actions. Thus, 
\[
Rd(\sigma^{(n)}) = (m-1) \alpha.
\]

The CoRadius of $\sigma^{(n)}$ is calculated from the minimum resistance path from $\sigma^{(0)}$ to $\sigma^{(n)}$ as shown in Fig. \ref{fig:2DGrid}. This path comprises a transition from the configuration in Fig. \ref{fig:2DGrid}(a) to (b) with a resistance of one. The first node experiences a resistance of $1/2$, the next two neighboring nodes experience a resistance of $1/4$, and the fourth node transitions with zero resistance.  Then, a transition from Fig. \ref{fig:2DGrid}(b) to (c) that requires five nodes to select noisy actions for $m = 7$ and $m-2$ noisy actions in general. Of these $m-2$ players selecting noisy actions, the two players on the boundary will face a resistance of $1/3$, and the remaining $m-4$  internal players will have a resistance of $1/4$. Finally, the transition from Fig. \ref{fig:2DGrid}(c) to (d) requires one noisy action per row for the remaining $m-2$ rows with zero participating players. Within these $m-2$ players, again, the internal $m-4$ players will have a resistance of $1/4$, and the two boundary players will have a resistance of $1/3$. Combining all these resistances, we get
\[
CR(\sigma^{(n)}) = 1 + 2 \frac{(m-4)}{4} + 2 \left(\frac{1}{3} + \frac{1}{3}\right) = \frac{7}{3} + \frac{1}{2} (m-4)
\]
for $m \geq 4$. If 
\begin{equation}\label{eq:RdCR_grid_upper}
    \alpha > \alpha_{\te{th}}^u = \frac{1}{m-1}\left(\frac{7}{3} + \frac{1}{4} (m-4)\right),
\end{equation}
then $\sigma^{(n)}$ is stochastically stable. 

Finally, we derive the condition for $\sigma^{(0)}$ to be stochastically stable. The Radius of $\sigma^{(0)}$ is the path resistance from the configuration in Fig. \ref{fig:2DGrid}(a) to \ref{fig:2DGrid}(b), which is equal to
$
Rd(\sigma^{(0)}) = 1.
$
For the CoRadius of $\sigma^{(0)}$, Fig. \ref{fig:CR_2DGrid} presents a minimum resistance path from $\sigma^{(n)}$ to $\sigma^{(0)}$ with resistance
\begin{equation}\label{eq:CR_grid}
CR(\sigma^{(0)}) = (m-1)\alpha + 2\alpha\left( \sum_{y\in \{3,5,\ldots,m-2\}}(m-y) \right).
\end{equation}
Thus, if 
\begin{equation}\label{eq:RdCR_grid_lower}
    \alpha < \alpha_{\te{th}}^{l} = CR(\sigma^{(0)}),
\end{equation}
where $CR(\sigma^{(0)})$ is defined in (\ref{eq:CR_grid}), then $\sigma^{(0)}$ is stochastically stable. We want to reiterate that the $Rd-CR$ analysis yields sufficient conditions for a stochastically stable profile. Therefore, for $\alpha_{\te{th}}^l < \alpha < \alpha^u_{\te{th}}$, $\sigma^{(0)}$ and $\sigma^{(n)}$ can still be stochastically stable. 
\subsection{Simulations}
We conducted extensive simulations to validate the stochastic stability results for wheel and grid networks. The line network has a unique Nash equilibrium in which no users participate. In the ring network, there are two Nash equilibria. However, the threshold condition on $\alpha$ is directly proportional to the number of agents in the network. For instance, even for a small network of $n = 10$ agents, $\alpha$ must exceed $5$ for $\sigma^{(n)}$ to be stochastically stable. Given that the results for these two configurations are straightforward, we did not include their simulation results.

For the wheel and grid networks, we validated the conditions on $\alpha$ as outlined in Eqs. (\ref{eq:RdCR_wheel}), (\ref{eq:RdCR_grid_upper}), and (\ref{eq:RdCR_grid_lower}). The simulation results in Fig. \ref{fig:RDCR_wheel_grid} support our theoretical predictions. We simulated a wheel network with $n = 20$ players under two different scenarios: one where $\alpha$ was set to $\alpha_{\te{th}} + \epsilon$ and the other where $\alpha$ was set to $\alpha_{\te{th}} - \epsilon$, with $\alpha_{\te{th}}$ representing the threshold value derived from Eq. (\ref{eq:RdCR_wheel}), and $\epsilon = 0.1$. Each scenario was run 100 times, and the average results are displayed in Fig. \ref{fig:RDCR_wheel_grid} (left). The figure clearly shows that when $\alpha$ exceeds the threshold value, full participation by all players becomes stochastically stable, whereas when $\alpha$ falls below the threshold, non-participation becomes the stochastically stable outcome.

We applied the same analysis on a $5 \times 5$ grid, setting $\epsilon = 0.01$ and the noise parameter $T = 0.1$. The results, depicted in Fig. \ref{fig:RDCR_wheel_grid} (right), demonstrate that when $\alpha = \alpha_{\te{th}}^u + \epsilon$, the Nash equilibrium $\sigma^{(n)}$, where all players participate, is stochastically stable. Conversely, when $\alpha = \alpha_{\te{th}}^l - \epsilon$, the Nash equilibrium $\sigma^{(0)}$, where no player participates, is stochastically stable. These findings validate the theoretical predictions regarding the stability of network configurations under varying conditions of $\alpha$. 

What happens for the values of $\alpha$ between $\alpha_{\te{th}}^l$ and $\alpha_{\te{th}}^u$? Given that the network participation game has multiple Nash equilibria in a grid network, could other equilibria be stochastically stable under certain conditions? To explore this question, we simulated the grid network for intermediate values of $\alpha$. As noted earlier, the $Rd$-$CR$ analysis only provides a sufficient condition for stochastic stability, leaving open the possibility that the same equilibria ($\sigma^{(0)}$ and $\sigma^{(n)}$) could be stable at these intermediate $\alpha$ values. We considered the following values of $\alpha$: $(0.125,~0.225,~0.325,~0.425,~0.525,~0.625)$, where $\alpha_{\te{th}}^l = 0.125$ and $\alpha_{\te{th}}^u = 0.625$. The results, shown in Fig. \ref{fig:grid_combined}, reveal that only the extreme profiles, $\sigma^{(0)}$ and $\sigma^{(n)}$, are stochastically stable, even for intermediate values of $\alpha$. Specifically, $\sigma^{(0)}$ remains stochastically stable for $\alpha = 0.125$ and $0.225$. For $\alpha = 0.325$, network participation oscillated between $\sigma^{(0)}$ and $\sigma^{(n)}$. However, for $\alpha=0.425$ and onwards, $\sigma^{(n)}$ was stochastically stable. Our simulations indicate that the remaining Nash equilibria are not stochastically stable at the intermediate values of $\alpha$ as well.
\begin{figure}
    \centering
    \includegraphics[width=0.46\linewidth]{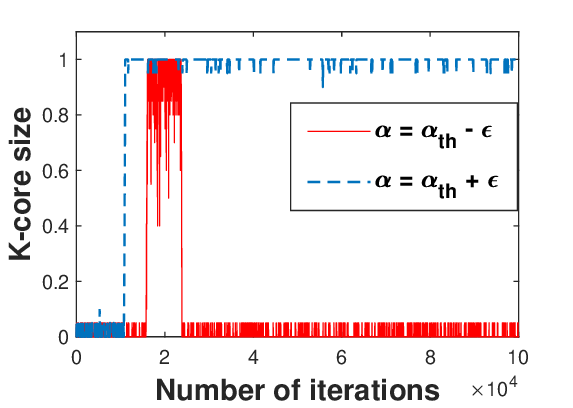}
    \includegraphics[width=0.46\linewidth]{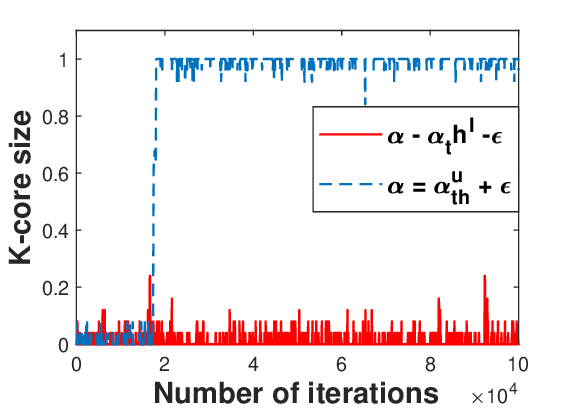}
    \caption{Simulation results of $Rd-CR$ analysis of Nework participation game for wheel (left) and grid (right) networks. For the wheel network (left), $n = 20$, $T = 0.042$, and $\epsilon = 0.1$. For the $5\times5$ grid network, $T = 0.1$ and $\epsilon = 0.01$.}
    \label{fig:RDCR_wheel_grid}
\end{figure}
\begin{figure*}[t]
    \centering
    \subfigure[$\alpha = 0.125$]{
        \includegraphics[scale=0.31]{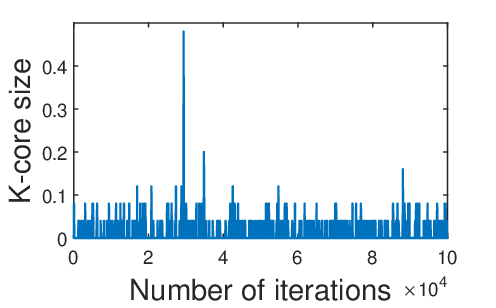} 
        \label{fig:grid_1}
    }   
    \subfigure[$\alpha = 0.225$]{
        \includegraphics[scale=0.31]{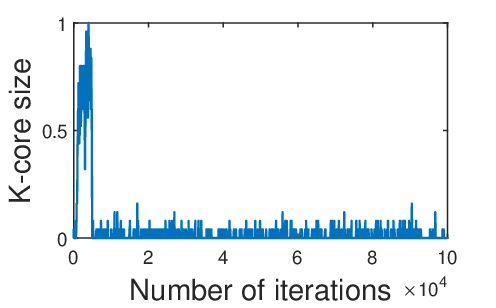} 
        \label{fig:grid_1}
    }    
    \subfigure[$\alpha = 0.325$]{
        \includegraphics[scale=0.31]{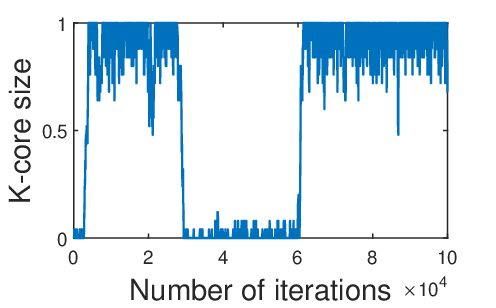} 
        \label{fig:grid_1}
    }      
    \subfigure[$\alpha = 0.425$]{
        \includegraphics[scale=0.31]{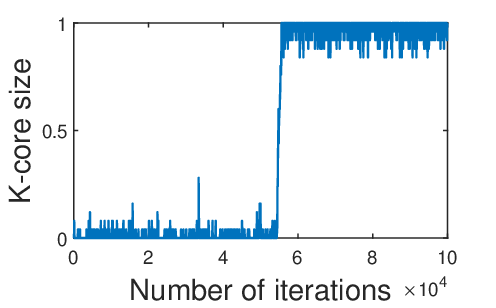} 
        \label{fig:grid_1}
    }   
    \subfigure[$\alpha = 0.525$]{
        \includegraphics[scale=0.31]{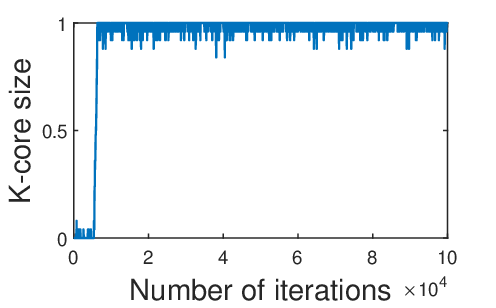} 
        \label{fig:grid_1}
    }   
    \subfigure[$\alpha = 0.625$]{
        \includegraphics[scale=0.31]{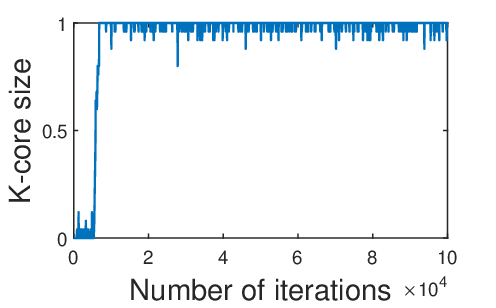} 
        \label{fig:grid_1}
    }   
    \caption{Simulation results for grid networks for $\alpha_{\te{th}}^l \leq \alpha \leq \alpha_{\te{th}}^u$}
    \label{fig:grid_combined}
    \end{figure*}
\section{Game Design to Maximize User Engagement}\label{sec:InfluenceMaximization}
In the preceding sections, we introduced a game-theoretic framework to model user engagement in network participation and resource-sharing problems. We represented players as boundedly rational decision-makers and illustrated that our framework captured the diffusion and non-diffusion dynamics of network evolution. 
\emph{In this section, we approach this setup from a system designer's perspective, aiming to maximize user participation in the network activity}. 
The standard approach in the literature involves selecting anchor nodes; players incentivized to participate in network activities regardless of their payoffs. The problem of computing an optimal set of anchor nodes is an NP-hard combinatorial optimization problem. Therefore, several greedy heuristics or simulated annealing-based algorithms have been proposed to solve this problem. \emph{We propose a real-time approach for selecting anchor nodes and quantifying each node's impact on system performance, which can be used to compute the incentives for that node. } 
%
\subsection{Network Participation: System Design Perspective}
For the network participation game, we present the following optimization problem to maximize user participation:
\begin{equation}\label{eq:opt_problem}
    \max_{\sigma \in \mc{A}}~\Phi^{\te{NPG}} (\sigma) = \max_{\sigma \in \mc{A}}~ \sum_{i = 1}^n U_i^{\te{NPG}}(\sigma_{i},\sigma_{-i}),
\end{equation}
where $U_i^{\te{NPG}}(\sigma)$ represents the utility that players receive from participating in cooperative activities. Players disengage if the number of their participating neighbors falls below the threshold value of $k$. Therefore, to maximize $\Phi^{\te{NPG}}(\sigma)$, all users with at least $k$ participating neighbors must engage. Conversely, players with fewer than $k$ neighbors will contribute negatively to the global objective if they participate. However, if such players can help others exceed their participation threshold, the cost of their participation will be outweighed by the net benefit they contribute to the global objective. Incentivizing these players can further boost overall user participation in the network.

\begin{thm}\label{thm:MCU}
Consider the network participation game with $n$ players. We define the utility function of each player as
\begin{equation}\label{eq:U_i_global}
U^{g,\te{NPG}}_i(\sigma_i,\sigma_{-i}) = 
\left\lbrace
\begin{array}{cc}
 U^{\te{NPG}}_i(\sigma) + I^{\te{NPG}}_i(\sigma)    &   \sigma_i = 1, \\
 0 & \sigma_i = 0,
\end{array}
\right.
\end{equation}
where $U^{\te{NPG}}_i(\sigma)$ is defined in (\ref{eq:Util-k}) and $I^{\te{NPG}}_i(\sigma)$ is 
\begin{equation}\label{eq:I_i}
    I^{\te{NPG}}_i(1,\sigma_{-i}) = \sum_{j \in \mc{N}^{p,\te{th}}_i(\sigma)} \left(\alpha_j + \frac{1}{|\mc{N}_j|}\right) +\sum_{j \in \mc{N}^{p,\te{ex}}_i(\sigma)}  \frac{1}{|\mc{N}_j|}. 
\end{equation}
Here, 
\begin{align*}
\mc{N}^{p,\te{th}}_i(\sigma) &= \{j \in \mc{N}^p_i(\sigma)~|~ |\mc{N}^p_j(\sigma)| = k\}, \text{ and }\\
\mc{N}^{p,\te{ex}}_i(\sigma) &= \{l \in \mc{N}^p_i(\sigma)~|~ |\mc{N}^p_l(\sigma)| < k\}.
\end{align*}
Then, the resulting game is an exact potential game with a potential function $\Phi^{\te{NPG}}(\sigma)$ in (\ref{eq:opt_problem}). Moreover, if all the players update their actions using Log-Linear Learning, then the set of stochastically stable action profiles ${\te{SS}}^{\te{NPG}}$ is the set of  action profiles where potential function is maximum, i.e., 
\[
{\te{SS}^{\te{NPG}}} = \argmax_{\sigma \in \mc{A}} \sum_{i=1}^n U^{\te{NPG}}_i.
\]
\end{thm}
\begin{proof}
We start with the global function and derive the utility of individual players using the Marginal Contribution Utility (MCU) approach presented in \cite{wolpert1999introduction}. Then, the resulting game will be an exact potential game with the global objective function as a potential function for that game (see the details in \cite{Marden12}). 

From the global objective function defined in (\ref{eq:opt_problem}), the Marginal Contribution Utility of a player is 
\[
U^{\te{MCU}}_i(\sigma_i,\sigma_{-i}) = \Phi^{\te{NPG}}(\sigma_i,\sigma_{-i}) - \Phi^{\te{NPG}}(\sigma^{\te{bl}}_i,\sigma_{-i}),
\]
where $\sigma^{\te{bl}}_i$ is a baseline action for $i$. Thus, the MCU of player $i$ for an action $\sigma_i$, given the actions of other players $\sigma_{-i}$, is the difference in the global function when player $i$ switches from its baseline action to $\sigma_i$. For the network participation game, we consider $\sigma^{\te{bl}}_i = 0$ for all the players, i.e., the baseline action for a player is not to participate. Then, the utility of a player for participating in the network activity is 
\[
U^{\te{MCU}}_i(1,\sigma_{-i}) = \Phi^{\te{NPG}}(1,\sigma_{-i}) -\Phi^{\te{NPG}}(0,\sigma_{-i}). 
\]
The first term on the right-hand side is the global function when player $i$ is participating, and the second term is the function when $i$ is not participating. We get the following expression by replacing the definition of the potential function $\Phi^{\te{NPG}}$ from (\ref{eq:opt_problem}) in the above equation. 
\begin{align*}
    U^{\te{MCU}}_i( 1, \sigma_{-i})&= \sum_{j=1}^n\left[U^{\te{NPG}}_j(1, \sigma_{-i}) - U^{\te{NPG}}_j(0, \sigma_{-i}) \right], \\
    &= [U^{\te{NPG}}_i(1,\sigma_{-i}) - U^{\te{NPG}}_i(0,\sigma_{-i}] + 
     \\
    \sum_{ j \in \mc{N}^p_i(\sigma)}& [U^{\te{NPG}}_j(1,\sigma_{-i}) - U^{\te{NPG}}_j(0,\sigma_{-i})].
\end{align*}
For the second equality above, we used the facts that player $i$'s actions do not impact the non-neighboring players of $i$, and $U^{\te{NPG}}_j(0,\sigma_{-i}) = 0$. Thus, we only have to consider the impact of $i$'s actions on the participating neighbors of $i$. 

We can divide the players in $\mc{N}^p_i(\sigma)$ into three categories. Players for which the number of participating neighbors is at least one higher than the threshold value. Such players will not be impacted by player $i$'s decisions because their participation threshold is satisfied even without player $i$. The second category of players in $\mc{N}^p_i(\sigma)$ consists of those players for which the number of participating neighbors is exactly equal to their threshold. These players will suffer a loss in utility without player $i$. Let $\mc{N}^{p,\te{th}}(\sigma)$ be the subset of such participating players. Thus, for all $j \in \mc{N}^{p,\te{th}}(\sigma)$
\[
  U^{\te{NPG}}_j(1,\sigma_{-i}) - U^{\te{NPG}}_j(0,\sigma_{-i}) =  \left(\alpha + \frac{1}{|\mc{N}_j|}\right). 
\]
The third category of players in $\mc{N}^p_i(\sigma)$ is the set of players whose number of participating neighbors is less than their threshold. Their decision to participate was exploratory, and they incurred a negative utility. Let $\mc{N}^{p,\te{ex}}_i(\sigma)$ be the subset of such players. Then, for all $l \in \mc{N}^{p,\te{ex}}(\sigma)$
\[
U^{\te{NPG}}_l(1,\sigma_{-i}) - U^{\te{NPG}}_l(0,\sigma_{-i}) =  \frac{1}{|\mc{N}_l|} .
\]
Combining the above information results in the following utility for player $i$. 
\begin{equation*}
    U^{\te{MCU}}_i(1,\sigma_{-i}) = U^{\te{NPG}}_i(1,\sigma_{-i}) + I^{\te{NPG}}_i(1,\sigma_{-i})
\end{equation*}
where $I^{\te{NPG}}_i(1,\sigma_{-i})$ is as defined in (\ref{eq:U_i_global}).
%
Thus, the global utility proposed in (\ref{eq:U_i_global}) is the Marginal Contribution Utility derived in the above expression. 

Since $U^{g,\te{NPG}}_i(\sigma)$ is derived from $\Phi^{\te{NPG}}(\sigma)$ according to the Marginal Contribution Utility approach, the global function $\Phi^{\te{NPG}}(\sigma)$ is the potential function for the resulting game by definition as shown in \cite{marden2013distributed}. Finally, it was proved in \cite{Marden12} and \cite{alos10} that for exact potential games, if players update their actions according to Log-Linear Learning, then the set of stochastically stable profiles is the same as the set of potential maximizing action profiles, which concludes the proof. 
 \end{proof}
 
The first component of the updated utility function $U^{g, \te{NPG}}_i(1,\sigma_{-i})$ in (\ref{eq:U_i_global}) is $U^{\te{NPG}}_i(1,\sigma_{-i})$, which is the payoff received by player $i$ for participating in the cooperative activity. The second component is $I^{\te{NPG}}_i(1,\sigma_{-i})$, which represents the impact of player $i$ on the payoffs of its neighbors. 
The function $I^{\te{NPG}}_i(\sigma)$ quantifies the contribution of $i$ towards maximizing user participation in the network activity. 
\begin{itemize}
\item The first term in $I^{\te{NPG}}_i(1,\sigma_{-i})$  represents the number of participating neighbors of $i$ that will stop participating if player $i$ leaves, which can initiate a cascade of withdrawals.  Therefore, this term is a direct measure of the significance of player $i$ in enhancing the network activity. 
\item The second term in $I_i(1,\sigma_{-i})$  reflects the impact of  $i$ in supporting neighbor's exploratory behavior. If player $i$ decides to leave, each of the neighboring players who were exploring to start participating will be negatively impacted, and this impact is quantified in the second term. 
\end{itemize} 

The utility function in Theorem \ref{thm:MCU} derived using the Marginal Contribution Utility, along with the resulting potential game, forms the foundation of our real-time anchor selection strategy. By leveraging this game setup, we present a novel approach that enables a system designer to select anchor nodes in real time. Specifically, the function $I_i^{\te{NPG}}(\sigma)$ quantifies a player's potential influence on their neighbors' decisions, providing crucial information for determining whether a player should be incentivized as an anchor node.
\subsection{Anchor Node Selection}
We introduce a Principal Agent (PA) into our framework whose objective is to drive the system towards an optimal solution to the optimization problem in (\ref{eq:opt_problem}) by incentivizing a subset of players. The action set of the PA is  $\mc{A}^{\te{PA}} = \{0,1\}^n$, representing a set of $n$-dimensional vectors with binary entries. For any vector $\sigma^{\te{PA},\te{NPG}} \in \mc{A}^{\te{PA}}$, the element $\sigma^{\te{PA},\te{NPG}}_i$ denotes the status of player $i$, where $\sigma^{\te{PA},\te{NPG}}_i = 1$ indicates that the Principal Agent has selected player $i$ as an anchor node. The nodes designated as anchors are incentivized by the PA to participate based on their contributions to network activity. We incorporate this incentivization factor into the players' utility functions as follows:
\begin{equation*}\label{eq:Util_PA}
U^{\te{PA},\te{NPG}}_i(\sigma) = U^{\te{NPG}}_i(\sigma) + \sigma^{\te{PA},\te{NPG}}_i I^{\te{NPG}}_i(\sigma).    
\end{equation*}

The utility that players receive from their impact on neighbors is now linked to their anchor status, represented by $\sigma^{\te{PA},\te{NPG}}_i$. If the PA selects player $i$ as an anchor node, then $\sigma^{\te{PA},\te{NPG}}_i = 1$, and the player receives additional utility for participating. In contrast, for a regular player, where $\sigma^{\te{PA},\te{NPG}}_i = 0$, player $i$'s utility is solely determined by $U^{\te{NPG}}_i(\sigma)$. Thus, 
\begin{equation}\label{eq:U_PA_NPG}
  U^{\te{PA},\te{NPG}}_i(\sigma_i,\sigma_{-i}) = 
\left\lbrace
\begin{array}{cc}
 U^{g,\te{NPG}}_i(\sigma)    &\sigma^{\te{PA},\te{NPG}}_i = 1, \\
  U^{\te{NPG}}_i(\sigma)   & \sigma^{\te{PA},\te{NPG}}_i = 0.  
\end{array}
\right.  
\end{equation}
Here $U^{\te{NPG}}_i(\sigma)$ and $U^{g,\te{NPG}}_i(\sigma)$ are player utilities defined in (\ref{eq:U_PA_NPG}) and (\ref{eq:U_i_global}), respectively. This setup raises the question of PA's strategy for selecting anchor nodes. We present two methods for the Principal Agent to select anchor nodes. 
\subsubsection{Unconstrained Anchor Selection}
In this scenario, the PA has no restrictions on the number of anchor nodes that can be incentivized. Initially, the PA sets $\sigma^{\te{PA},\te{NPG}}_i = 0$ for all players $i \in {1, 2, \ldots, n}$, indicating that no players are incentivized at the outset. The system follows LLL where, at each iteration, $t$, a single player $i$ is randomly chosen for action revision, while all other players maintain their actions from the previous iteration, $t-1$. Based on the following criteria, the PA decides whether to offer anchor node status to player $i$. 
\begin{equation}\label{eq:PA_decision_unconstrained}
    \sigma^{\te{PA},\te{NPG}}_i(\sigma) = 
\left\lbrace
\begin{array}{ll}
 1    & U^{\te{NPG}}_i(\sigma) \leq 0 \text{ and } U^{g,\te{NPG}}(\sigma)> 0, \\
  0   & \text{Otherwise}.  
\end{array}
\right.
\end{equation}
Thus, a player is selected as an anchor node if the payoff $U^{\te{NPG}}_i(\sigma)$ is negative but their influence $I^{\te{NPG}}_i(\sigma)$ on neighbors is significant, leading to a positive value of $U_i^{g,\te{NPG}}(\sigma)$.

Once the PA determines the status of player $i$, the player decides participating in the network activity using the probabilistic decision rule of LLL, described in (\ref{eq:pi_LLL}), with the utility function $U_i^{\te{PA},\te{NPG}}(\sigma)$. If the PA sets $\sigma^{\te{PA},\te{NPG}}_i = 1$, but the player opts not to participate due to exploratory or noisy action, the PA resets $\sigma^{\te{PA},\te{NPG}}_i$ to zero. This adjustment conserves PA's budget, incentivizing only actively participating nodes.

\emph{A key feature of our approach is the dynamic reevaluation of anchor nodes}. Suppose that at some time $t_1$, if player $i$'s benefit was below the threshold but they contributed to increasing network participation, the PA designated them as an anchor. However, at a later time $t_2 > t_1$ when player $i$ receives another revision opportunity, if $i$'s participating neighbors increase, making $U^{\te{NPG}}_i(\sigma) > 0$, the PA will reset $\sigma^{\te{PA},\te{NPG}}_i$ to zero. Similarly, if player $i$ decides to disengage at time $t_2$ as a noisy action, the anchor status will be revoked. 

\subsubsection{Budget Constrained Anchor Selection}\label{subsection:budgetNPG}
In this scenario, the PA is constrained by a budget on the number of anchor nodes. Let $b$ be the maximum number of anchor nodes the PA can select. Then,  the PA's restricted action set is $
\mc{A}^{\te{PA}} = \{\sigma^{\te{PA},\te{NPG}} \in \{0,1\}^n ~|~ n_a \leq b\},
$
where $n_a = \sum_{i=1}^n \sigma^{\te{PA},\te{NPG}}_i$ is the number of active anchor nodes at that time. Initially, no player is incentivized. The PA maintains a dynamic set $\Gamma_t$, where each entry is a tuple $(i, I_i^{\te{NPG}}(\sigma(t-1)), t_{i}^a)$, with $t_i^a$ representing the number of intervals player $i$ has served as an anchor node. We propose that each anchor should receive the incentives $I_i^{\te{NPG}}$ for at least $t^{\te{th}}$ intervals. Let $\Gamma_t^{\geq,\te{th}}$ be the subset of $\Gamma_t$  consisting of anchors who have been receiving the incentive for more than $t^{\te{th}}$ intervals at time interval $t$, i.e., $t^a_i \geq ^{\te{th}}$.  

At each iteration $t$, when a random player $i$ receives a revision opportunity, the PA begins by evaluating the criteria outlined in (\ref{eq:PA_decision_unconstrained}) to determine whether the player qualifies as an anchor candidate. If the player meets the criteria, the PA checks the budget constraint. If the current number of active anchor nodes, $n_a$, is below the budget limit $b$, the PA assigns $\sigma^{\te{PA},\te{NPG}}_i = 1$, selecting player $i$ as an anchor. However, if $n_a \geq b$, meaning the maximum number of anchor nodes is already active, the PA assesses the potential contribution of player $i$, $I_i^{\te{NPG}}(\sigma(t-1))$, and compares it to the contributions of the current anchors in the list $\Gamma_t^{\geq,\te{th}}$. This list contains anchor nodes that have been receiving incentives for more than $t^{\te{th}}$ intervals. If player $i$'s contribution exceeds the lowest contribution among the anchors in $\Gamma_t^{\geq,\te{th}}$, then $\sigma^{\te{PA},\te{NPG}}_i = 1$, and player $i$ replaces the anchor with the smallest contribution. The PA also regularly reassesses the active anchor nodes' status in the set $\Gamma_t^{\geq,\te{th}}$ after every $t^u$ intervals. If any anchor $j$ in $\Gamma_t^{\geq,\te{th}}$ satisfies $U_j^{\te{NPG}} > 0$, the PA will revoke the anchor status for that player and will set $\sigma^{\te{PA},\te{NPG}}_j = 0$.  

After the PA's decision, the selected player $i$ chooses an action from their action set $\mc{A}_i$ based on the probability distribution defined by LLL process, as given in (\ref{eq:pi_LLL}). If the PA has previously set $\sigma^{\te{PA},\te{NPG}}_i(\sigma) = 1$, indicating that the player has been incentivized to engage, but the player instead selects a noisy action that results in not participating in the network activity, the PA will reset $\sigma^{\te{PA},\te{NPG}}_i(\sigma) = 0$, effectively withdrawing the incentive.

Thus, the proposed anchor selection strategy enables the PA to select anchor nodes in real time for unconstrained and budget-constrained scenarios. \emph{Another strength of our proposed strategy is a relaxed information requirement compared to the existing methods in \cite{bhawalkar2015preventing} and \cite{abbas2019graph}}. Computing a set of anchor nodes using greedy heuristics or simulated annealing before the start of the activity requires the PA to have complete knowledge of the underlying social network. However, in our approach, the PA only needs to know the neighborhood structure of player $i$ selected for action revision at time $t$ and the neighborhood information of $i$'s participating neighbors. Thus, the information requirements in our framework are more realistic and relaxed than those in the prior approaches. 
\subsection{Network Resource Sharing: System Design Perspective}
We extend our approach to maximizing user engagement in the resource-sharing problem setup. We follow the same procedure of starting with a global optimization problem and formulating the resource-sharing problem as a potential game. Once we have the potential game setup, the PA's strategies will be the same for the unconstrained and budget-constrained scenarios as presented for the network participation problem. 

In the resource-sharing problem. we propose the following global optimization problem:
\begin{equation}\label{eq:opt_prob_NSG}
    \max_{\sigma \in \mc{A}} \Phi ^{\te{NSG}}(\sigma) = \max_{\sigma \in \mc{A}} \sum_{i = 1}^n U_i^{\te{NSG}} (\sigma_i,\sigma_{-i}), 
\end{equation}
where $U_i^{\te{NSG}}(\sigma)$ is the player's utility for participating in the resource-sharing activity, as presented in (\ref{eq:U_rs}). Then, 
the Marginal Contribution Utility can be computed as follows:
\[
U^{g,\te{NSG}}_i(\sigma) =  \Phi ^{\te{NSG}}(1, \sigma_{-i}) - \Phi ^{\te{NSG}}(0,\sigma_{-i}).
\]
The revised utility function has the following form:
\begin{align}\label{eq:Ug_NSG}
    U^{g,\te{NSG}}_i(\sigma) &= U^{\te{NSG}}_i(\sigma) +I^{\te{NSG}}_i(\sigma), \\
    I^{\te{NSG}}_i(1,\sigma_{-i}) &= \sum_{j \in \mc{N}^{p,\te{th}}_i(\sigma)} \left(\alpha +  \frac{c_{ji}}{|L_i|}\right) +\sum_{j \in \mc{N}^{p,\te{ex}}_i(\sigma)}  \frac{c_{ji}}{|L_i|} \nonumber
\end{align}
with $c_{ji} = |L^p_j(1,\sigma_{-i})| - |L^p_j(0,\sigma_{-i})|$ and $U^{\te{NSG}}_i(\sigma)$ is player utility function defined in (\ref{eq:U_rs}) . Here, $|L^p_j(\sigma)|$ is defined in (\ref{eq:L}) and is equal to the number of unique resources that player $j$ can access. Therefore, $c_{ji}$ is the number of unique resources that $j$ can only access through $i$. 

In the above expression for $I^{\te{NSG}}_i(\sigma)$, the set $\mc{N}^{p,\te{th}}_i(\sigma)$ is the subset of neighbors of $i$ whose participation threshold is not satisfied without $i$. These players will suffer a loss of $\alpha + c_{ji}/r$ if $i$ disengages from the network. Similarly, the players in the set $\mc{N}^{p,\te{ex}}_i(\sigma)$ are those neighbors of $i$ who do not have access to the full set of $r$ resources. However, they still engage in the sharing activity as an exploratory behavior. Every member $j \in \mc{N}^{p,\te{ex}}_i(\sigma)$  will experience a loss of $c_{ji}/r$, where $c_{ji}$ is the number of resources that $j$ could access through $i$ only. Thus, the resource-sharing game with utility function in (\ref{eq:Ug_NSG}) is a potential game. Moreover, if players update their actions using LLL, then the stocahstically stable profile is the set of profiles that maximize the potential function $\Phi^{\te{NSG}}$ in (\ref{eq:opt_prob_NSG}). 

Next, we present the anchor node selection algorithm for the resource-sharing game. Since the utility function $U_i^{g,\te{NSG}}(\sigma)$ has the same structure as the $U_i^{\te{g,NPG}}(\sigma)$, we propose the same decision strategy for the Principal Agent (PA). We start with the unconstrained scenario in which there are no restrictions on the number of anchor nodes the PA can select. When a player is randomly selected to revise their strategy, the PA follows the following decision strategy. 
\begin{equation}\label{eq:PA_NSG}
    \sigma^{\te{PA},\te{NSG}}_i(\sigma) = 
\left\lbrace
\begin{array}{ll}
 1    & U^{\te{NSG}}_i(\sigma) \leq 0 \text{ and } U^{g,\te{NSG}},(\sigma)> 0, \\
  0   & \text{Otherwise}.  
\end{array}
\right.
\end{equation}
Similarly, for the unconstrained scenario, the PA will follow the same steps as outlined for the participation game. 
Thus, the utility function of player $i$ will be
\begin{equation}\label{eq:U_PA_NSG}
U^{\te{PA},\te{NSG}}_i(\sigma_i,\sigma_{-i}) = 
\left\lbrace
\begin{array}{cc}
 U^{g,\te{NSG}}_i(\sigma)    &\sigma^{\te{PA},\te{NSG}}_i = 1, \\
  U^{\te{NSG}}_i(\sigma)   & \sigma^{\te{PA},\te{NSG}}_i = 0.  
\end{array}
\right.    
\end{equation}
%
\begin{figure*}[t]
\centering
    \subfigure[$k = 7$, Network size = 781 nodes]{
    \includegraphics[width=0.32\textwidth]{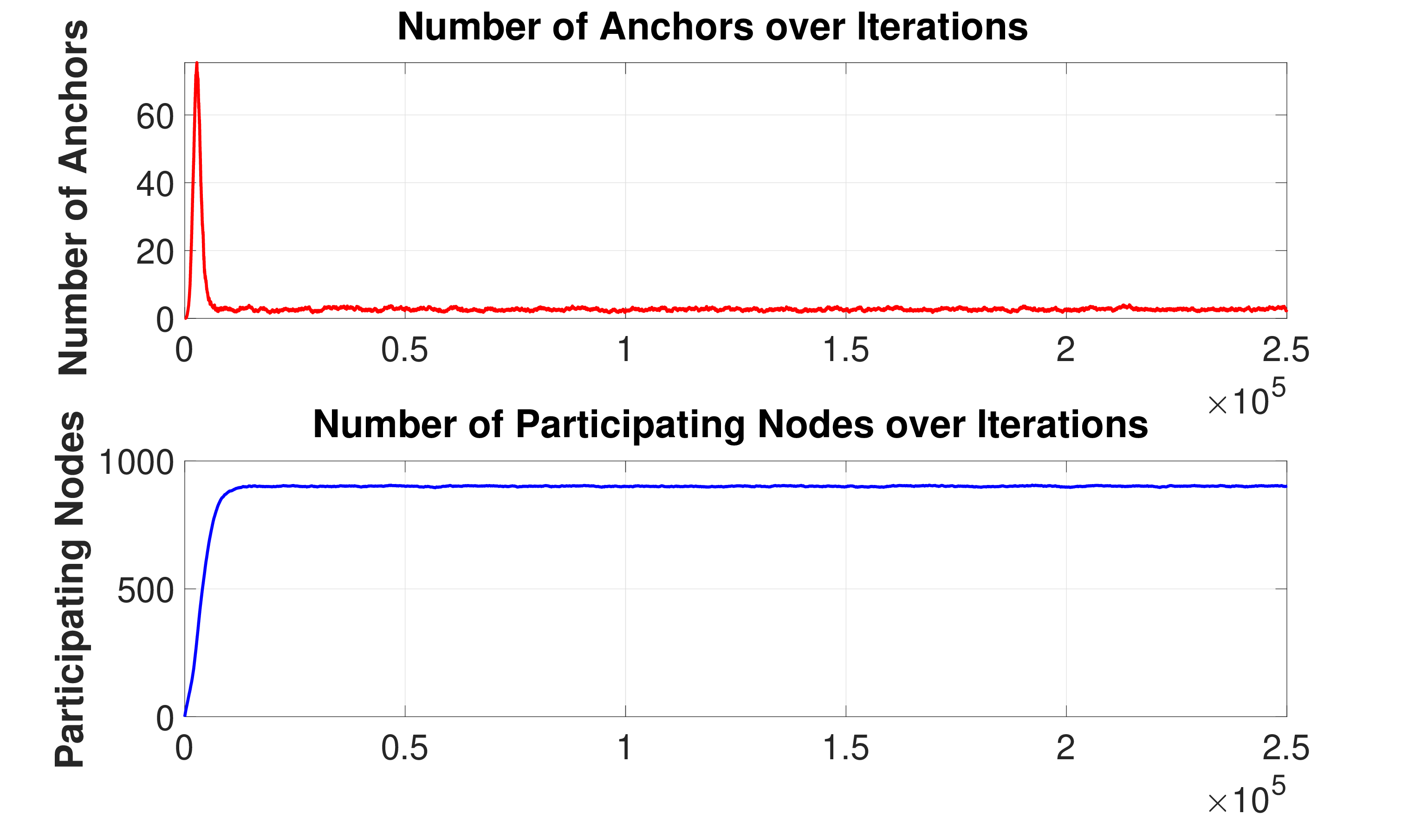}
    \label{fig:unconstrained_k=7}
    }
    \subfigure[$k = 8$, Network size = 687 nodes]{\includegraphics[width=0.32\textwidth]{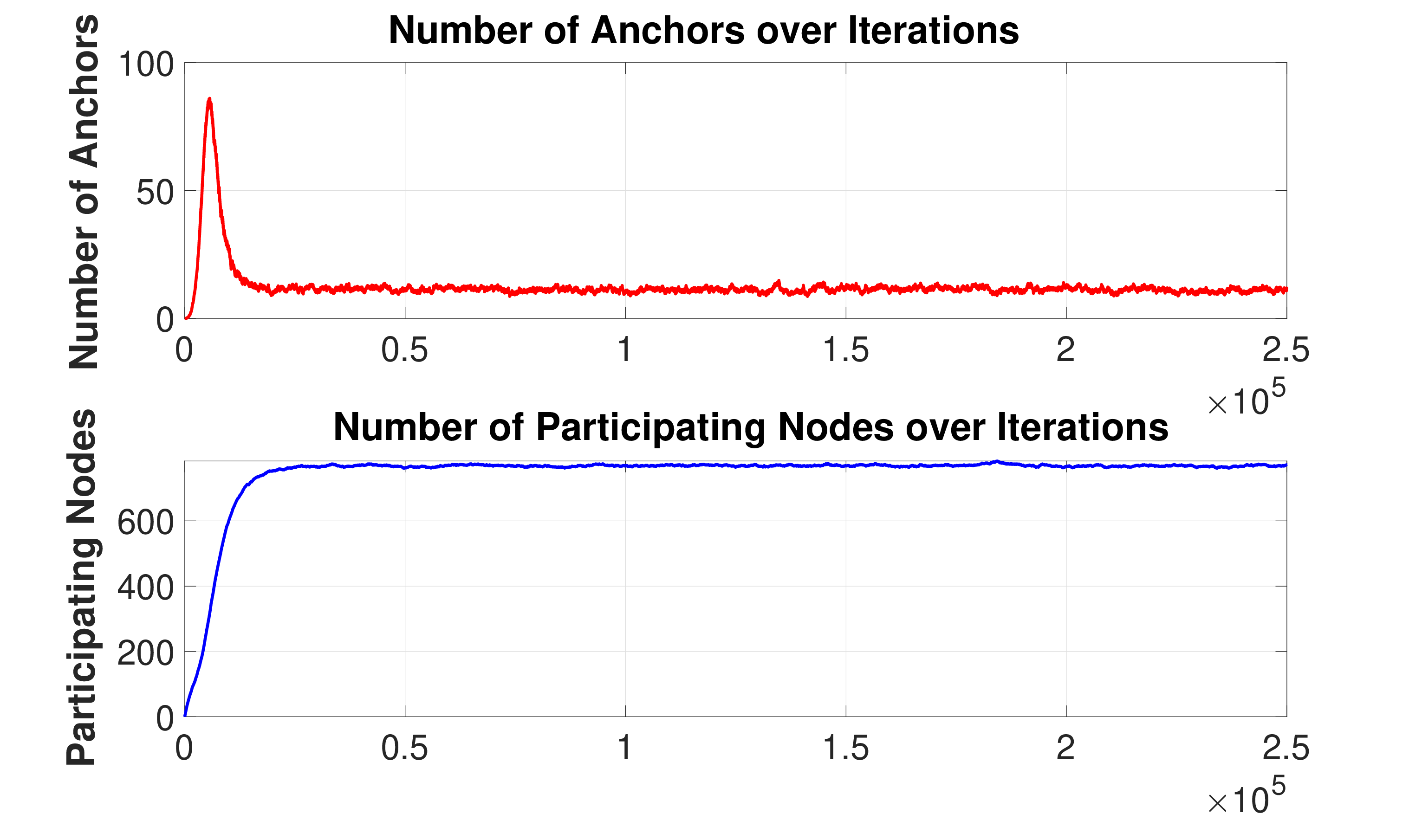}
    \label{fig:unconstrained_k=8}
    }
        \subfigure[$k = 9$, Network size = 354 nodes]{\includegraphics[width=0.32\textwidth]{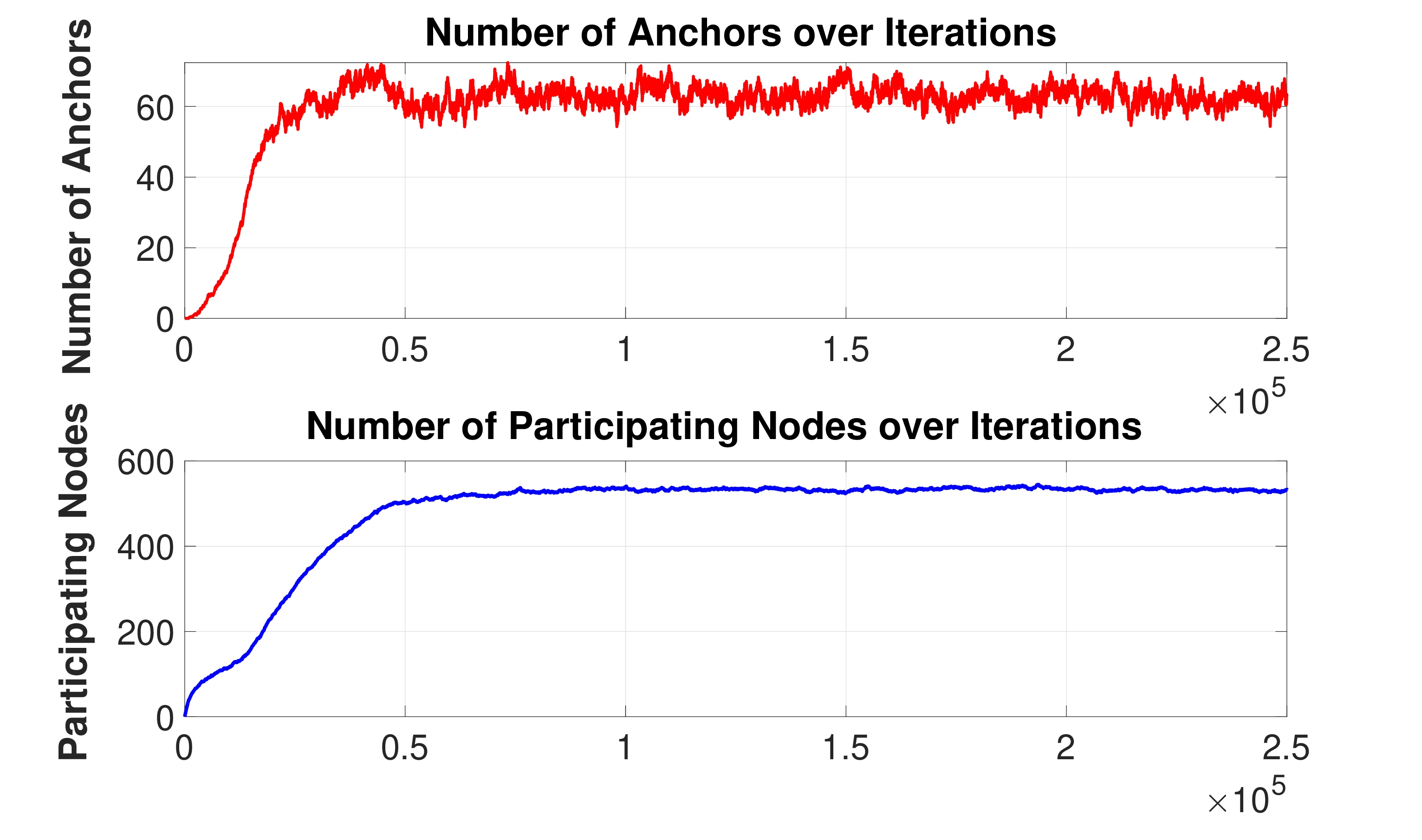}
        \label{fig:unconstrained_k=9}
        }  
        \subfigure[$k = 7$]{
    \includegraphics[width=0.32\textwidth]{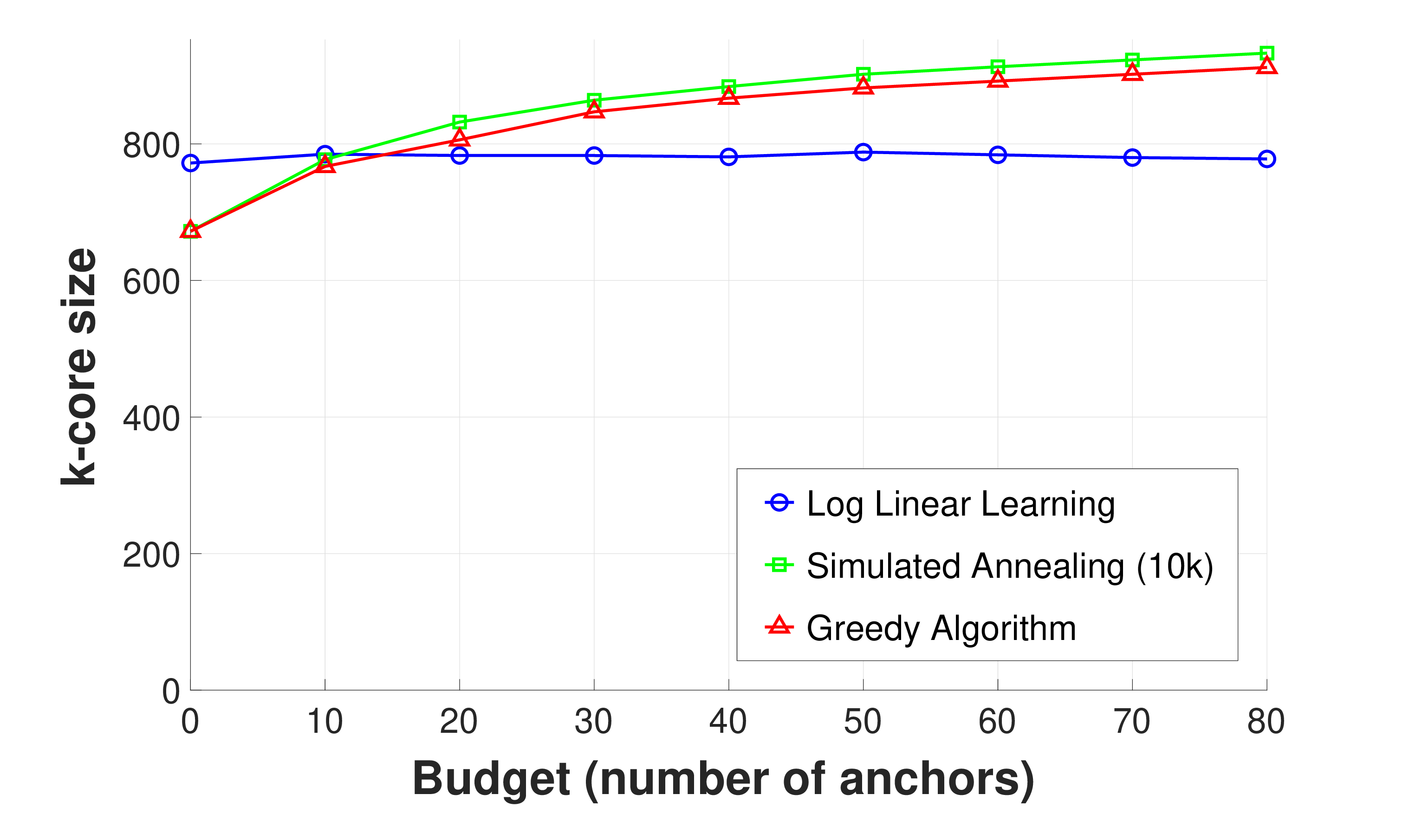}
    \label{fig:constrained_k=7}
    }
    \subfigure[$k = 8$]{\includegraphics[width=0.32\textwidth]{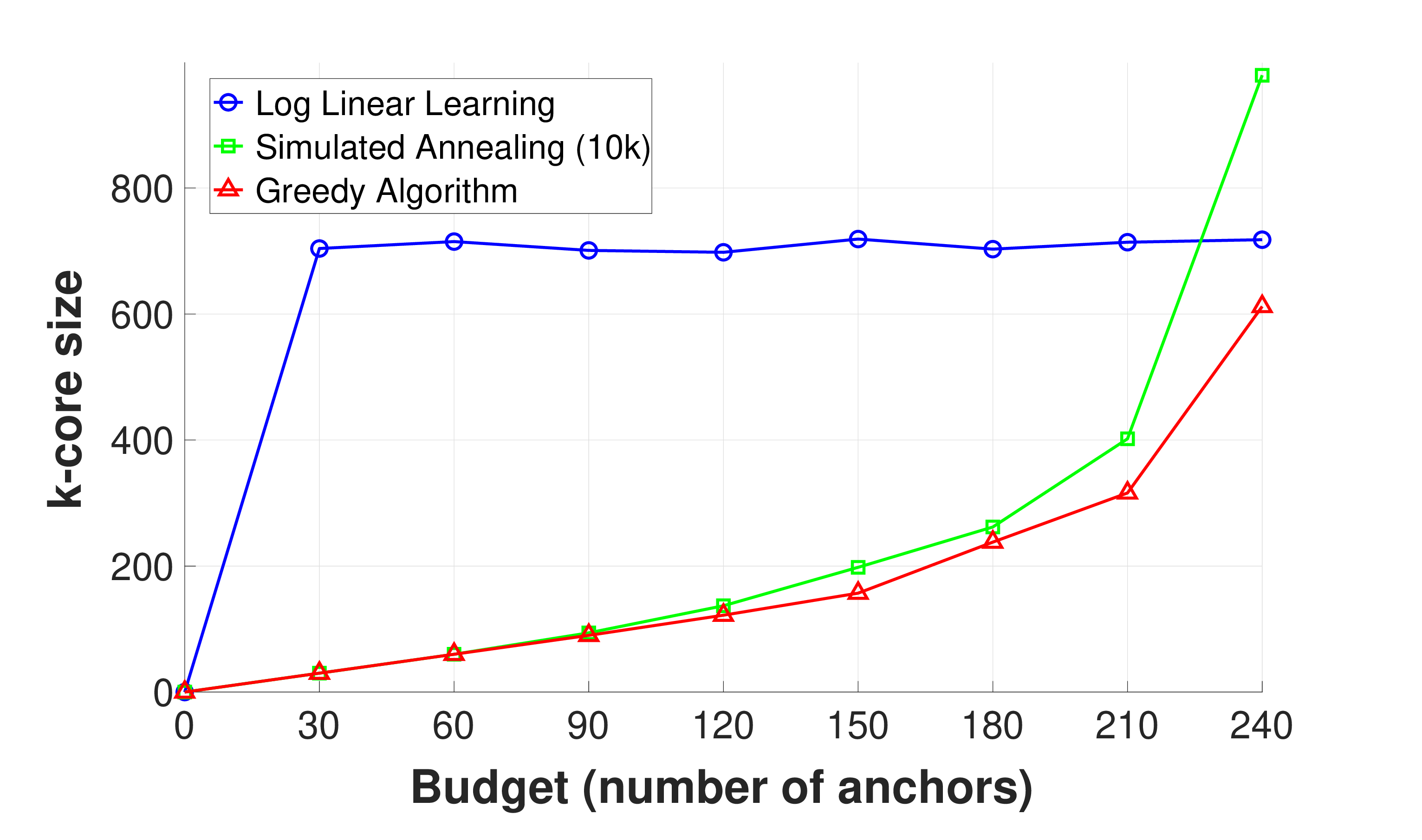}
    \label{fig:constrained_k=8}
    }
        \subfigure[$k = 9$]{\includegraphics[width=0.32\textwidth]{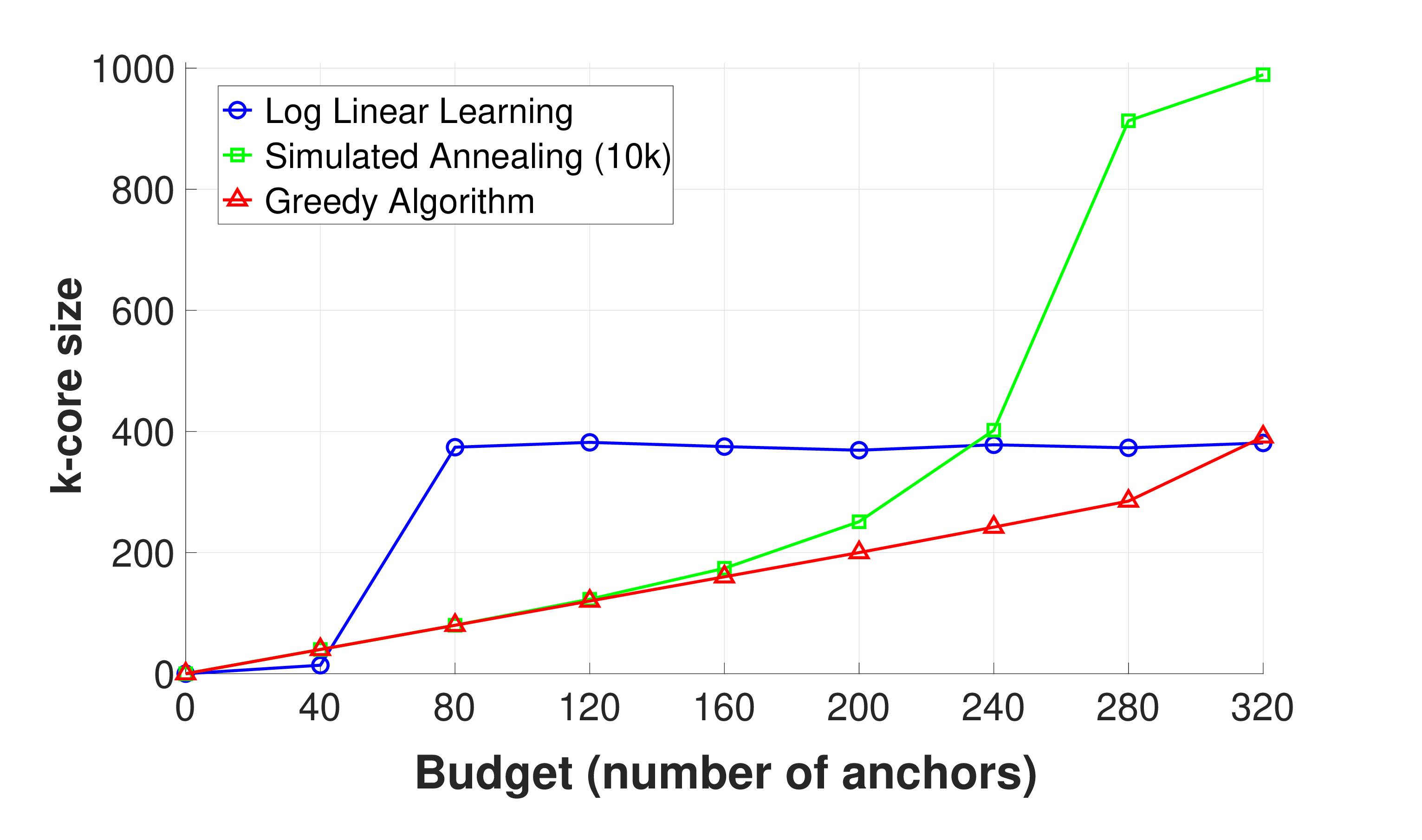}
    \label{fig:constrained_k=9}
    }
\caption{Simulation results for anchor selection strategy for network participation problem. The social network is an ER graph with $n = 1000$ and $p = 0.01$. The top plots show the total number of participating nodes and the bottom plots show the anchor nodes selected according to our proposed unconstrained framework.}
\label{fig:Unconstrained_NPG}
\end{figure*}
\vspace{-0.2in}
\subsection{Analysis}
We analyze the stochastic stability of the updated network participation and resource-sharing game setups with the Principal Agent. Our real-time anchor selection strategies for the network participation and resource-sharing problems are based on the potential game frameworks for the global optimization problems in (\ref{eq:opt_problem}) and (\ref{eq:opt_prob_NSG}). For potential games with utility function $U_i^{g,\te{NPG}}$ and $U_i^{g,\te{NSG}}$ derived using the Marginal Contribution Utility approach, we know that the stochastically stable network configurations are the ones that maximize the corresponding potential functions $\Phi^{\te{NPG}}$ and $\Phi^{\te{NSG}}$, respectively. 
However, the PA's actions updated the players' utility functions. Specifically, players with positive payoffs, i.e., $U_i^{\te{NPS}}(\sigma) > 0$ and $U^{\te{NSG}}_i(\sigma) > 0$, do not receive the additional credit for their contributions to their neighbors represented by $I_i^{\te{NPG}}$ and $I_i^{\te{NSG}}$. Consequently, the resulting games with utility functions $U_i^{\te{PA},\te{NPG}}$ and  $U_i^{\te{PA},\te{NSG}}$ are no longer exact potential games. 

\begin{prop}\label{prop:unconstrained_SS}
Consider the unconstrained versions of network participation and resource-sharing games with utility functions $U_i^{\te{PA},\te{NPG}}$ and  $U_i^{\te{PA},\te{NSG}}$ defined in Eqs. (\ref{eq:U_PA_NPG}) and (\ref{eq:U_PA_NSG}), respectively. 
\begin{enumerate}
    \item The games are best response potential games with potential functions $\Phi^{\te{NPG}}$ and $\Phi^{\te{NSG}}$ defined in Eqs. (\ref{eq:opt_problem}) and (\ref{eq:opt_prob_NSG}), respectively. 
    \item If all the players adhere to LLL, then the sets of stochastically stable configurations belong to the sets of Nash equilibria for these games. 
\end{enumerate}

\end{prop}
\begin{proof}

We know that a game is a best response potential game if the condition in (\ref{eq:BR_Potential}) is satisfied. Therefore, we must show that the following conditions are satisfied for all players $i$ and for all action profiles $\sigma \in \mc{A}$. 
\begin{align*}
    \argmax_{\sigma_i \in \mc{A}_i}U^{\te{PA},\te{NPG}}_i(\sigma) &=\argmax_{\sigma_i \in \mc{A}_i}\Phi^{\te{NPG}}_i(\sigma), \text{ and }\\
    \argmax_{\sigma_i \in \mc{A}_i}U^{\te{PA},\te{NSG}}_i(\sigma) &=\argmax_{\sigma_i \in \mc{A}_i}\Phi^{\te{NSG}}_i(\sigma).
\end{align*}
Since each player has only two actions, it is straightforward to verify that the above conditions hold. 

Statement 2) in the proposition is a direct consequence of the best response potential game nature of these games. 
\end{proof}
\subsection{Simulations}
We evaluate the performance of the proposed anchor selection strategies through extensive simulations. Moreover, we compare the performance of our proposed scheme with the anchor selection strategies presented in \cite{bhawalkar2015preventing} and \cite{abbas2019graph} for the participation and resource-sharing problems, respectively.

\begin{figure*}[t]
\centering
\subfigure[$(r,s) = (10,3)$, Network size = 723 nodes]{
\includegraphics[width=0.31\textwidth]{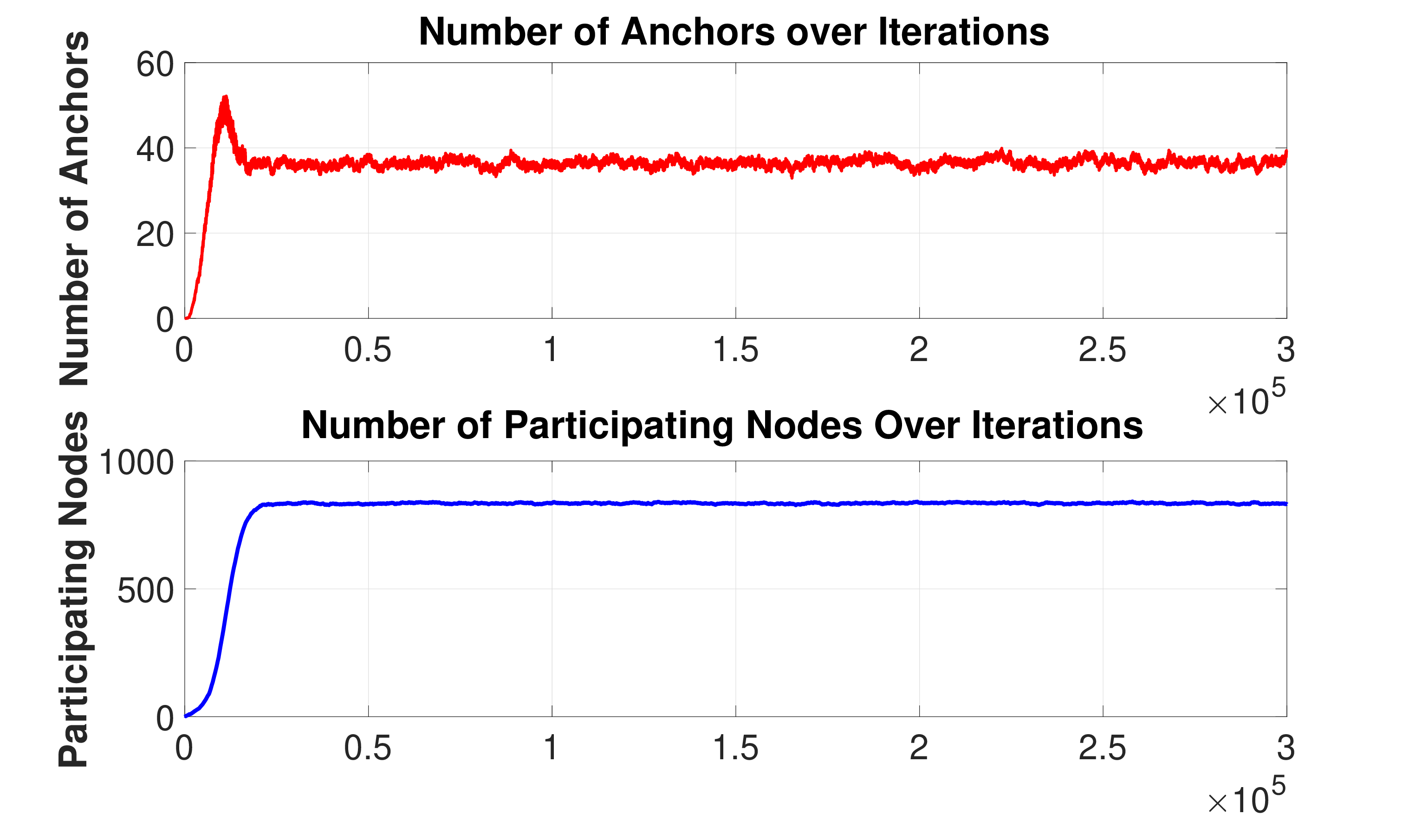}
\label{fig:unconstrained(r,s) = (10,3)}
}
    \subfigure[ $(r,s) = (6,2)$, Network size = 694 nodes]{
    \includegraphics[width=0.31\textwidth]{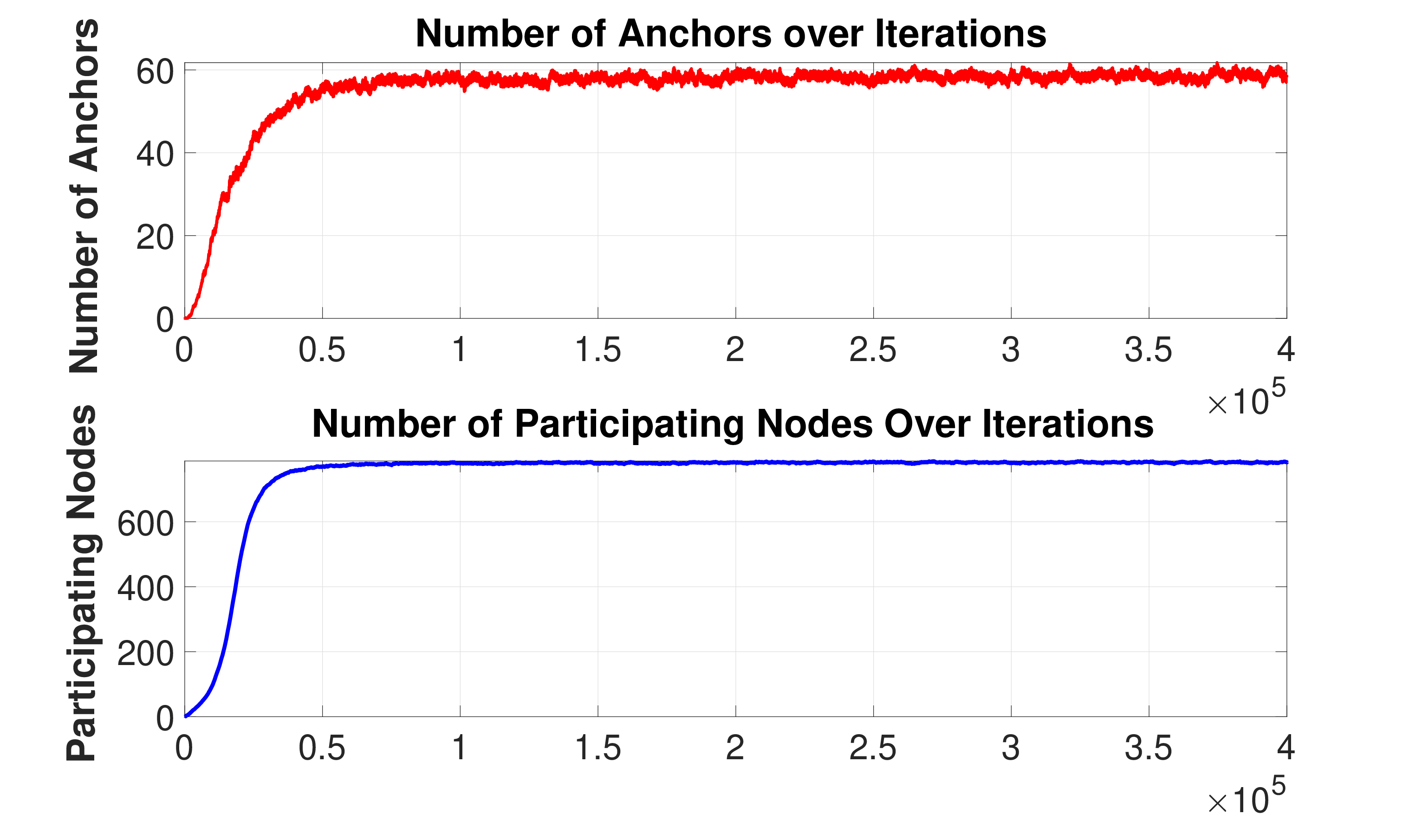}
    \label{fig:unconstrained(r,s) = (6,2)}
}
    \subfigure[ $(r,s) = (5,2)$, Network size = 742 nodes]{
    \includegraphics[width=0.31\textwidth]{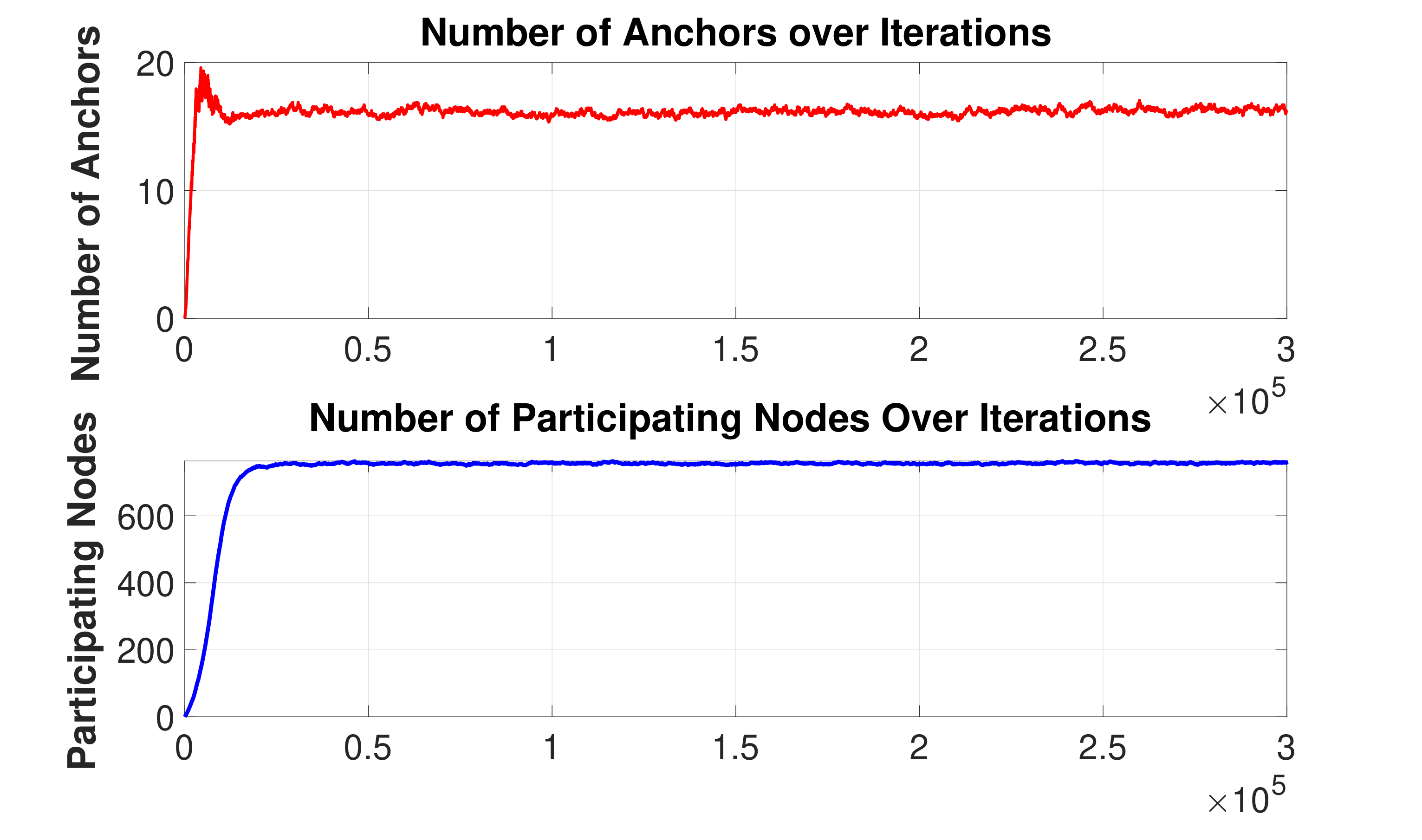}
    \label{fig:unconstrained(r,s) = (5,2)}
}
    \subfigure[$(r,s) = (10,3)$]
    {\includegraphics[width=0.31\textwidth]{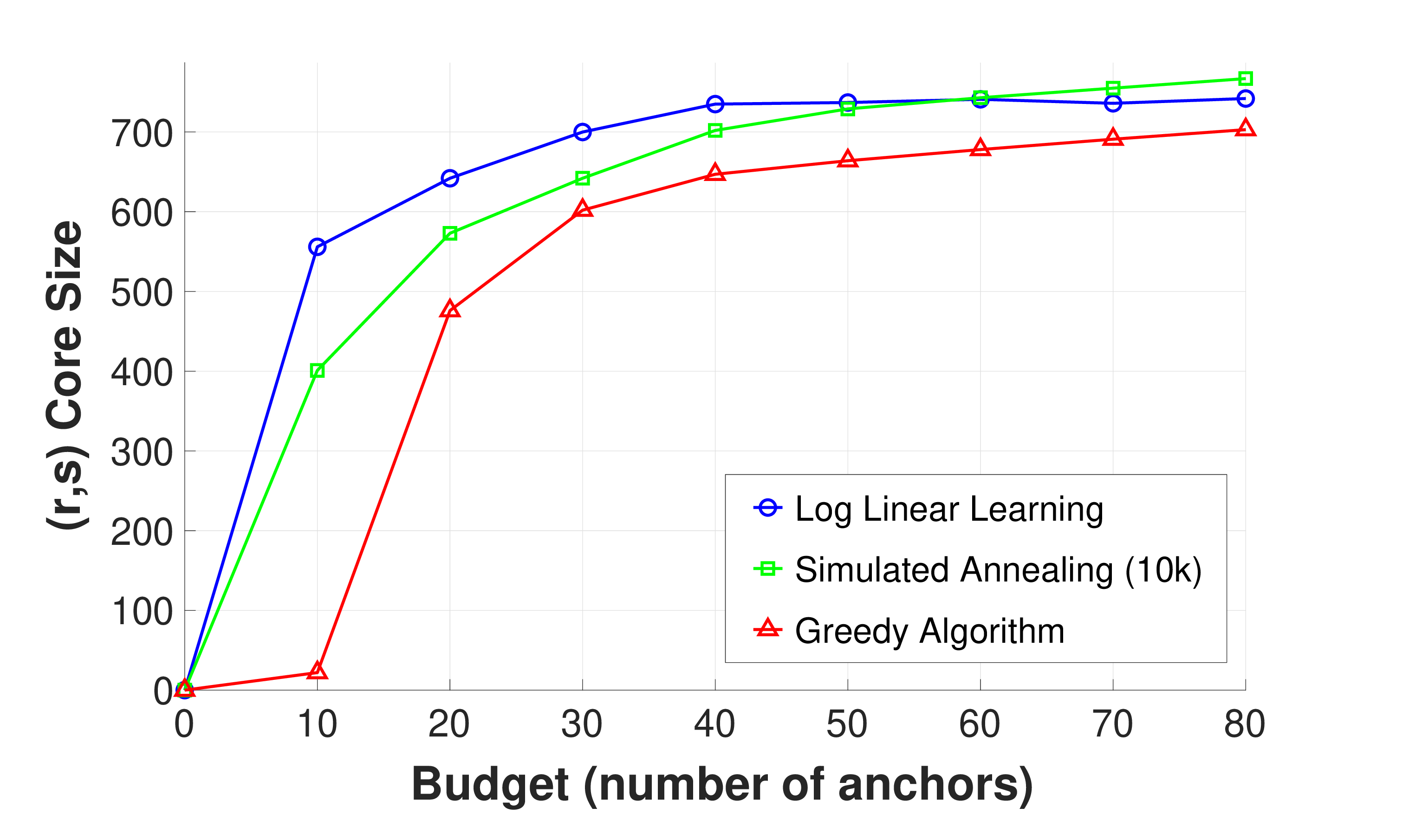}
    \label{fig:budget(r,s) = (10,3)}
    }
    \subfigure[$(r,s) = (6,2)$]{
    \includegraphics[width=0.31\textwidth]{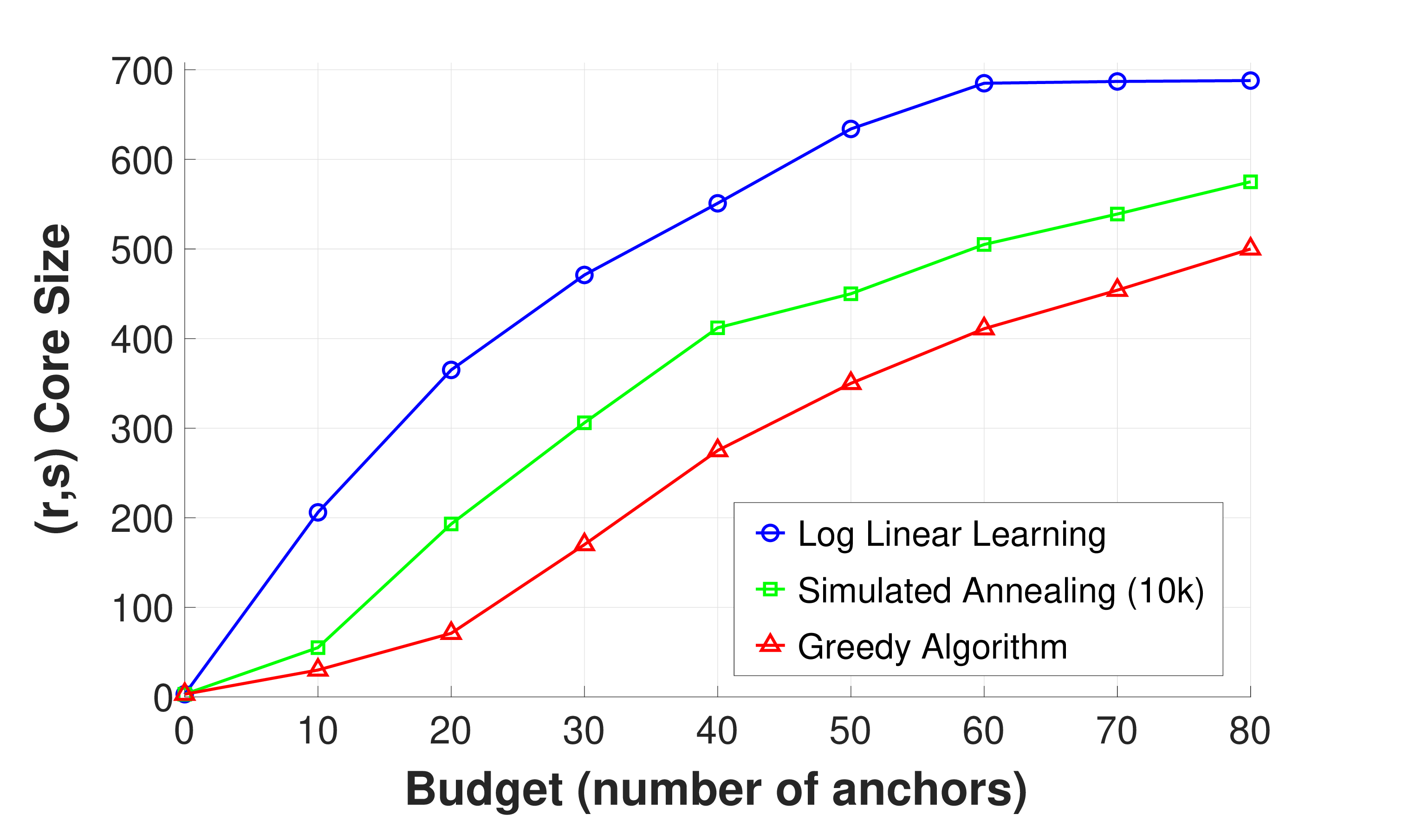}
    \label{fig:budget(r,s) = (6,2)}
    }
    \subfigure[$(r,s) = (5,2)$]{
    \includegraphics[width=0.31\textwidth]{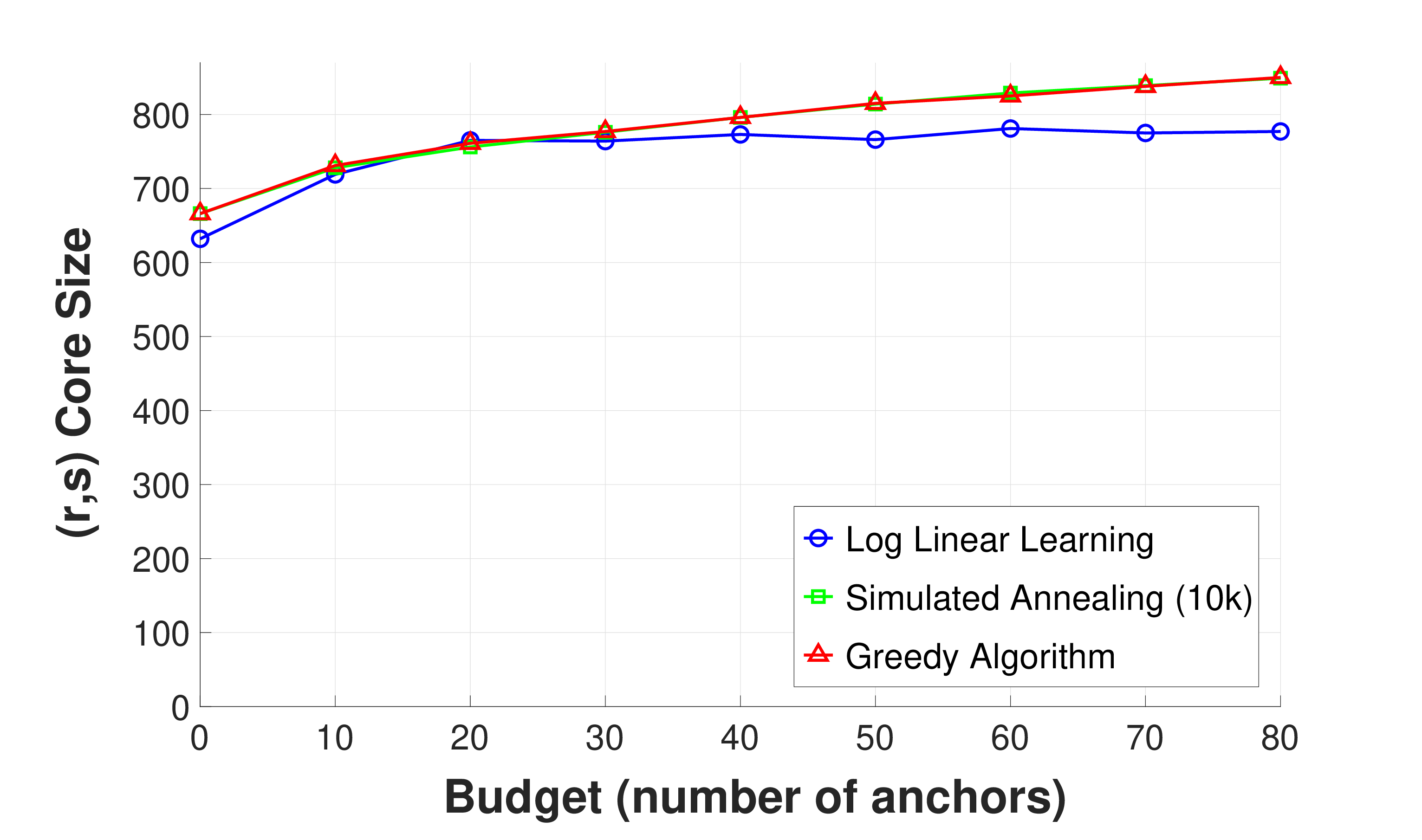}
    \label{fig:budget(r,s) = (5,2)}
    }    
\caption{Simulation results for the anchor selection strategy for the resource sharing game. The social network is an ER graph with $n = 1000$, $p = 0.01.$ The top plots in Figs (a)-(c) represent the number of anchor nodes, and the bottom plots represent the number of participating nodes for the unconstrained case. The plots in Figs. (d)-(f) represents a comparison of the $(r,s)$-core sizes under our proposed approach with the existing approaches. }
\label{fig:anchor_rs}   
\end{figure*}

We begin by analyzing the network participation game without constraints, as shown in Figs. \ref{fig:unconstrained_k=7}-\ref{fig:unconstrained_k=9}. The top plots display the number of anchor nodes, while the bottom plots show the evolution of player participation. The simulation setup follows Section \ref{sec:Simulation_local} using an ER graph with 1000 players and a link probability of $p = 0.01$, yielding an average of 10 neighbors per player. Players are considered participants if they remain active for more than 90\% of the final 100,000 iterations. For the budget-constrained scenario, the update interval is $t^u = 100$ with a minimum anchor duration of $t^{\text{th}} = 1000$ intervals. In Fig. \ref{fig:LLL_kcore}, it was shown that for $k \in \{5,6\}$, the $k$-core includes nearly the entire network. For $k = 7$, around 80\% of players are in the $k$-core, but for $k \geq 8$, the $k$-core is empty with zero initial participation. Under random initial conditions, 40\% of players were part of the $k$-core for $k = 8$, while the $k$-core remained empty for $k = 9$. Therefore, we focus on $k \in \{7,8,9\}$ to evaluate the effect of anchor nodes with zero initial participation.

The results in the bottom plots of Figs. \ref{fig:unconstrained_k=7}-\ref{fig:unconstrained_k=9} show that network participation improves markedly by adding a few anchor nodes. For $k = 9$, the network size increased from zero to around 350 players with just 60 anchor nodes. For $k = 8$, the network size increased from zero in the case of no initial participants and 270 in the case of a random number of initial participants to over 700 players. Participation also increased for $k = 7$, although the $k$-core size was already significant in this case. We want to highlight here that the number of participating players also includes the anchors.

The top plots in Figs. \ref{fig:unconstrained_k=7}-\ref{fig:unconstrained_k=9} provide interesting insights into the dynamics of our anchor selection strategy. Focusing on the top plot in Fig. \ref{fig:unconstrained_k=8}, we observe that the number of anchors increased rapidly in the initial phase, peaking at around eighty nodes, before gradually decreasing and stabilizing at a steady-state value of approximately twenty. One would expect this behavior when designing an advertising campaign for a new product, which, in our case, is a collaborative activity or an online community. Initially, with no users, the campaign requires aggressive incentives to attract participation. During this phase, the campaign designer must offer generous incentives to bring in as many users as possible. However, once the product gains traction within the community, the incentives can be allocated more strategically, allowing the designer to reduce the budget. Finally, as a significant portion of the community becomes engaged, the campaign can focus on a smaller, more targeted set of users.

We can observe a similar behavior in Fig. \ref{fig:unconstrained_k=7} for $k = 7$, where after an initial peak of eighty, the number of anchors converged to a value of less than ten. However, for $k = 9$, the behavior is different. The number of anchors steadily increased and settled at 60. This behavior is also justified because the threshold of $k =9$ is strict for a network with an average number of ten neighbors. However, the gains from this incentivization are also significant, increasing network size from zero to over 300 players. 

Figures \ref{fig:constrained_k=7}-\ref{fig:constrained_k=9} show the results of the budget-constrained participation game, where the participation network size is plotted against the anchor node budget, $b$. We compare the network size under our proposed framework to the anchored $k$-core generated using the greedy and simulated annealing algorithms from \cite{abbas2019graph}. \emph{In our approach, the final number of anchors in the unconstrained case establishes an upper bound on the number of anchors that can contribute to the constrained scenario}. Consequently, for $k = 8$, a rapid increase in network size occurs as the budget $b$ rises from zero to ten, but the network size stabilizes and remains constant with further increases in $b$. Similarly, for $k = 9$, the network size increases sharply when $b$ approaches 60, which is the number of anchors in the unconstrained case. However, the network size remains steady for $b > 60$. Thus, our approach cannot utilize the additional budget, and the number of anchor nodes selected in the constrained scenario does not increase beyond those in the unconstrained scenario. One possible solution to this limitation is to switch to a different selection strategy once the limit on the number of anchors is achieved. 

On the other hand, under the greedy and simulated annealing algorithms, the number of anchor nodes continues to increase as the budget rises. The plots indicate that in these algorithms, adding anchor nodes does not incentivize additional player participation for $b < 180$, leading to a participation network consisting solely of anchors. As we keep increasing the budget, the size of the participation network approaches that of the entire network under simulated annealing. To explain this behavior, assume that we have a sufficiently large budget for anchors. An optimal anchor selection strategy would be to identify players with less than $k$ neighbors and select them as anchors. If all such players are incentivized to participate as anchors, we will have a fully participating network. This is precisely what is happening in the case of simulated annealing. While full network participation appears attractive, the budget $b$ required to achieve this state is impractical. \emph{The results under our proposed framework appear closer to reality, in which we can gain a significant improvement in network participation by introducing a reasonable number of anchors. }

In the resource-sharing game, the simulation results are illustrated in Fig. \ref{fig:anchor_rs}. Figures \ref{fig:unconstrained(r,s) = (10,3)}-\ref{fig:unconstrained(r,s) = (5,2)} display the outcomes for the unconstrained scenario, while Figures \ref{fig:budget(r,s) = (10,3)}-\ref{fig:budget(r,s) = (5,2)} present results for the budget-constrained scenario. For the constrained case, we compare the size of the $(r,s)$-core obtained using our proposed approach to those derived from the greedy and simulated annealing algorithms in \cite{abbas2019graph}. As shown in Figures \ref{fig:budget(r,s) = (10,3)}-\ref{fig:budget(r,s) = (5,2)}, the $(r,s)$-core is empty without anchors in the $(10,3)$ and $(6,2)$ cases but includes approximately 750 players in the $(5,2)$ case. Furthermore, from Figures \ref{fig:unconstrained(r,s) = (10,3)} and \ref{fig:unconstrained(r,s) = (6,2)}, we observe that incentivizing forty and fifty players as anchors increases the participating network size from zero to around 700 players.

For the budgeted anchor selection problem, we compared our results with standard anchor selection methods, including simulated annealing and a greedy algorithm presented in \cite{abbas2019graph}. The plots in Figs. \ref{fig:budget(r,s) = (10,3)}-\ref{fig:budget(r,s) = (5,2)} demonstrate the increase in $(r,s)$-core size as the number of anchor nodes increases. These plots indicate that the expansion of the sharing network, when anchor nodes are selected using our real-time strategy, is comparable to that achieved with greedy and simulated annealing-based heuristics. Our results establish that the proposed approach performs on par with existing methods in terms of network size, while offering greater flexibility in anchor node selection and improved robustness to node failures.
  
Thus, from the simulation results presented in this section, we can conclude that the proposed anchor selection framework can effectively identify and incentivize players to maximize user engagement in network participation and resource-sharing scenarios.

\section{Conclusions}\label{sec:conclusion}
We proposed a game-theoretic formulation for network engagement problems in the context of network participation and resource-sharing problems.
\begin{itemize}
    \item We modeled the users in a social network as boundedly rational decision-makers, and we established that this behavioral model was instrumental in modeling the diffusion and non-diffusion aspects of network evolution proposed in multiple empirical studies of social networks. 
    \item We illustrated that our proposed framework enhanced the quality and rigor of steady-state analysis of the network engagement problem. We completely characterized the set of Nash equilibria for both the games. We proved that the network configuration will converge to the set of Nash equilibria if all the players adhere to LLL. Moreover,  we presented a stochastic stability analysis of the network participation problem for important classes of networks using the Radius-CoRadius results. 
    \item We presented a real-time approach for maximizing user engagement in network participation and resource-sharing problems. We proved that the games that were designed with a Principal-Agent belonged to the class of best response potential games.
    \item We verified the performance of our proposed framework through extensive simulations. In all the scenarios, we showed that our performance was comparable with the results in the existing literature. However, our framework offers significant enhancement in terms of the depth of system analysis. 
\end{itemize}
%
%
\bibliography{IEEEabrv,Robustness}
\end{document}